\documentclass[acmsmall]{acmart}

\renewcommand\footnotetextcopyrightpermission[1]{} 
\usepackage{tikz}
\usepackage{amsmath}
\usepackage{url}
\usepackage{listings}
\usepackage{tikz}
\usepackage{subfigure}
 
\usepackage[
	n,
	operators,
	advantage,
	sets,
	adversary,
	landau,
	probability,
	notions,
	logic,
	ff,
	mm,
	primitives,
	events,
	complexity,
	asymptotics,
	keys]{cryptocode}
\usepackage{amsmath}
\usepackage[]{cryptocode}
\usepackage{amsmath,amsfonts}

\usepackage[ruled]{algorithm2e}

\usepackage{graphicx}
\usepackage{textcomp}
\usepackage{xcolor}
\usepackage{multirow}
\usepackage{hyperref}
\usepackage{color}
\usepackage{booktabs}

\usepackage{wasysym}
\usepackage[flushleft]{threeparttable}
\usepackage{tablefootnote}
\usepackage{footnote}
\usepackage{subcaption}
 \usepackage{caption}
\newtheorem{theorem}{\bf Theorem}[section]

\newcommand{\para}[1]{\vspace{2pt}\noindent\textbf{#1}}
\usepackage{titlesec}
\titlespacing{\section}{0pt}{6.5pt}{6.5pt}
\titlespacing{\subsection}{0pt}{5pt}{5pt}

\usepackage{mathrsfs}
\newcounter{Lcount}
\newcommand{\squishlisttwo}{
\begin{list}{\arabic{Lcount}. }
{ \usecounter{Lcount}
\setlength{\itemsep}{0pt}
\setlength{\parsep}{0pt}
\setlength{\topsep}{0pt}
\setlength{\partopsep}{0pt}
\setlength{\leftmargin}{2em}
\setlength{\labelwidth}{1.5em}
\setlength{\labelsep}{0.5em} } }

\newcommand{\squishtwoend}{
\end{list} }

\newcommand{\squishlist}{
   \begin{list}{$\bullet$}
    { \setlength{\itemsep}{0pt}      \setlength{\parsep}{0pt}
      \setlength{\topsep}{3pt}       \setlength{\partopsep}{0pt}
      \setlength{\listparindent}{-2pt}
      \setlength{\itemindent}{-5pt}
      \setlength{\leftmargin}{0.5em} \setlength{\labelwidth}{0em}
      \setlength{\labelsep}{0.5em} } }

\newcommand{\squishend}{
    \end{list}  }

\usepackage{listings, xcolor}

\definecolor{verylightgray}{rgb}{.97,.97,.97}

\lstdefinelanguage{Solidity}{
	keywords=[1]{anonymous, assembly, assert, balance, break, call, callcode, case, catch, class, constant, continue, contract, debugger, default, delegatecall, delete, do, else, emit, event, export, external, false, finally, for, function, gas, if, implements, import, in, indexed, instanceof, interface, internal, is, length, library, log0, log1, log2, log3, log4, memory, modifier, new, payable, pragma, private, protected, public, pure, push, require, return, returns, revert, selfdestruct, send, storage, struct, suicide, super, switch, then, this, throw, transfer, true, try, typeof, using, value, view, while, with, addmod, ecrecover, keccak256, mulmod, ripemd160, sha256, sha3}, 
	keywordstyle=[1]\color{blue}\bfseries,
	keywords=[2]{address, bool, byte, bytes, bytes1, bytes2, bytes3, bytes4, bytes5, bytes6, bytes7, bytes8, bytes9, bytes10, bytes11, bytes12, bytes13, bytes14, bytes15, bytes16, bytes17, bytes18, bytes19, bytes20, bytes21, bytes22, bytes23, bytes24, bytes25, bytes26, bytes27, bytes28, bytes29, bytes30, bytes31, bytes32, enum, int, int8, int16, int24, int32, int40, int48, int56, int64, int72, int80, int88, int96, int104, int112, int120, int128, int136, int144, int152, int160, int168, int176, int184, int192, int200, int208, int216, int224, int232, int240, int248, int256, mapping, string, uint, uint8, uint16, uint24, uint32, uint40, uint48, uint56, uint64, uint72, uint80, uint88, uint96, uint104, uint112, uint120, uint128, uint136, uint144, uint152, uint160, uint168, uint176, uint184, uint192, uint200, uint208, uint216, uint224, uint232, uint240, uint248, uint256, var, void, ether, finney, szabo, wei, days, hours, minutes, seconds, weeks, years},	
	keywordstyle=[2]\color{teal}\bfseries,
	keywords=[3]{block, blockhash, coinbase, difficulty, gaslimit, number, timestamp, msg, data, gas, sender, sig, value, now, tx, gasprice, origin},	
	keywordstyle=[3]\color{violet}\bfseries,
	identifierstyle=\color{black},
	sensitive=false,
	comment=[l]{//},
	morecomment=[s]{/*}{*/},
	commentstyle=\color{gray}\ttfamily,
	stringstyle=\color{red}\ttfamily,
	morestring=[b]',
	morestring=[b]"
}

\AtBeginDocument{%
  \providecommand\BibTeX{{%
    \normalfont B\kern-0.5em{\scshape i\kern-0.25em b}\kern-0.8em\TeX}}}

\setcopyright{acmcopyright}
\copyrightyear{2018}
\acmYear{2018}
\acmDOI{XXXXXXX.XXXXXXX}

\acmConference[Conference acronym 'XX]{Make sure to enter the correct
  conference title from your rights confirmation emai}{June 03--05,
  2018}{Woodstock, NY}
\acmPrice{15.00}
\acmISBN{978-1-4503-XXXX-X/18/06}

\begin{document}

\title{opML: Optimistic Machine Learning on Blockchain}

\author{KD Conway}
\affiliation{%
  \institution{Hyper Oracle}
}
\email{0x1cc@hyperoracle.io}

\author{Cathie So}
\affiliation{%
  \institution{Hyper Oracle}
}
\email{cathie@hyperoracle.io}

\author{Xiaohang Yu}
\affiliation{%
  \institution{Hyper Oracle}
}
\email{nom4dv3@hyperoracle.io}

\author{Kartin Wong}
\affiliation{%
  \institution{Hyper Oracle}
}
\email{kartin@hyperoracle.io}


\begin{abstract}
The integration of machine learning with blockchain technology has witnessed increasing interest, driven by the vision of decentralized, secure, and transparent AI services. In this context, we introduce opML (Optimistic Machine Learning on chain), an innovative approach that empowers blockchain systems to conduct AI model inference. 
opML lies a interactive fraud proof protocol, reminiscent of the optimistic rollup systems. This mechanism ensures decentralized and verifiable consensus for ML services, enhancing trust and transparency. 
Unlike zkML (Zero-Knowledge Machine Learning), opML offers cost-efficient and highly efficient ML services, with minimal participation requirements. Remarkably, opML enables the execution of extensive language models, such as 7B-LLaMA, on standard PCs without GPUs, significantly expanding accessibility.
By combining the capabilities of blockchain and AI through opML, we embark on a transformative journey toward accessible, secure, and efficient on-chain machine learning.

\end{abstract}



\makeatletter \gdef\@ACM@checkaffil{} \makeatother 
\maketitle
\pagestyle{plain}

\section{Introduction}

In the rapidly evolving digital landscape, technological innovations are continually reshaping the way we interact with and harness the power of information. Among these innovations, the convergence of two remarkable forces, Artificial Intelligence (AI) and blockchain technology, stands out as a pivotal development. AI, with its capacity for advanced data analysis and decision-making, and blockchain, a decentralized ledger known for its security and transparency, have joined forces to explore new frontiers in the digital realm \cite{salah2019blockchain,dinh2018ai,yang2022fusing}.
As two pioneering forces, each with its distinct capabilities, AI and blockchain, are now merging to redefine the boundaries of what's possible in the digital world. This synergy has given rise to the concept of "Onchain AI" \cite{salah2019blockchain,dinh2018ai,yang2022fusing}, a paradigm that holds the promise of delivering decentralized, secure, and efficient AI services directly within the blockchain network.

However, a prevalent challenge in the current landscape of "Onchain AI" is the infeasibility of conducting AI computations directly on chain \cite{li2022smartvm}. 
For example, a simple task of naïve matrix multiplication of $1000 \times 1000$ integers would cost over 3 billion gas \cite{zheng2021agatha}, which far exceeds the current Ethereum’s block gas limit \cite{wood2014ethereum}.
Consequently, most of these applications resort to off-chain computations on centralized servers, only uploading the results onto the blockchain. 
While this strategy may yield functional results, it inherently sacrifices decentralization. Such a trade-off not only poses significant security challenges but also diminishes the core principles of trust and transparency that blockchain technology aims to uphold.

One alternative approach is to leverage Zero-Knowledge Machine Learning (zkML) \cite{zhang2020zero,weng2021mystique}. zkML represents a novel paradigm in the integration of machine learning and blockchain. 
zkML's reliance on zk-SNARKs (Zero-Knowledge Succinct Non-Interactive Arguments of Knowledge) \cite{goldreich1994definitions,ben2019scalable,ben2014succinct,bootle2016efficient} has been pivotal in safeguarding confidential model parameters and user data during the training and inference processes. This not only mitigates privacy concerns but also reduces the computational burden on the blockchain network, making zkML a promising candidate for decentralized ML applications.

While zkML undeniably presents a range of advantages in enhancing privacy and security within machine learning on the blockchain, it is crucial to acknowledge its inherent limitations. One of the most prominent challenges is the high cost associated with proof generation in zkML. The process demands considerable computational resources, resulting in extended generation times and substantial memory consumption. Consequently, zkML is most suitable for relatively small models, as its inefficiency becomes apparent when handling larger and more complex models. 
zkML would take more than 1000 times the memory and computation consumption for the ZK proof generation \cite{zkmodulus}.
As a result, zkML may prove impractical for extensive AI applications that require the processing of substantial datasets and intricate model parameters.

Faced with the limitations of zkML, we turn to investigate using fraud proof to prove the correctness of the ML results on chain instead of using the zero-knowledge proof (ZKP, also known as validity proof). 
Fraud proofs are commonly used in blockchain systems, including rollup systems, which are part of a broader category known as optimistic systems. Notable examples of rollup systems employing fraud proofs include Optimism \cite{optimism} and Arbitrum \cite{kalodner2018arbitrum}.
In the system using fraud proof, we optimistically assume that every proposed result is valid by default. In the case of mitigating potentially invalid results, the system using fraud proof will introduce a challenge period during which participants may challenge the submitter. The fraud proof is generated via interactive pinpoint protocol proving that the provided result is incorrect. The arbitration process is designed to validate a fraud proof by executing only a few minor computation steps, making the on-chain cost extremely low. 

Building upon the foundation of optimistic system design, we introduce opML: Optimistic Machine Learning on the blockchain\footnote{opML is open sourced. \url{https://github.com/hyperoracle/opml}}. Diverging from the approach of zkML, which relies on zero-knowledge proofs, opML adopts a fraud-proof system to guarantee the correctness of ML results. In the opML framework, submitters can execute ML computations within a native environment and subsequently provide the results directly on the blockchain. This approach maintains an optimistic assumption that each proposed result is inherently valid. During the challenge process, validators will check the correctness of these submitted result. 
If the result is invalid, he will start the dispute game (bisection protocol) with the submitter and tries to disprove the claim by pinpointing one concrete erroneous step. Finally, arbitration about a single step will be conducted on smart contract.

\begin{table}[!th]
\begin{tabular}{|l|l|l|}
\hline
\textbf{}      & \textbf{opML}                           & \textbf{zkML}               \\ \hline
Model size     & any size                                & small/limited               \\ \hline
Proof & fraud proof                             & validity proof  \\ \hline
Requirement    & Any PC with CPU/GPU                     & Large memory for proof generation \\ \hline
Finality       & Delay for challenge period              & No delays                   \\ \hline
Service cost   & low                                    & extremely high              \\ \hline
Security       & crypto-economic security (AnyTrust) & cryptographic security      \\ \hline
\end{tabular}
\caption{Comparison between opML and zkML}
\label{table:cmp_zk_op}
\end{table}

\paragraph{opML v.s. zkML.}
The comparison of opML and zkML is shown in Table \ref{table:cmp_zk_op}. Generally speaking, opML and zkML are different from the following perspectives:

\begin{itemize}
    \item \textbf{Proof System}: opML adopts the fraud proof, while zkML adopts the zero-knowledge proof to prove the correctness of the ML results.
    \item \textbf{Performance}: The performance is considered in the two perspectives: The proof generation time, and the memory consumption.
    \begin{itemize}
        \item The proof time of zkML is long, and significantly grows with increasing model parameters\footnote{An example of a significantly long proving time of zkML is running a nanoGPT model inference with 1 million parameters inside zkML framework takes around 80 minutes of proving time \cite{nanoGPT}.}. But for opML, computing the fraud proof can be conducted in the nearly-native environment in a short time.
        \item The memory consumption of zkML is extremely large compared to the native execution. But opML can be conducted in the nearly-native environment with the nearly-native memory consumption. Specifically, opML can handle a 7B-LLaMA model (the model size is around 26GB) within 32GB memory. But the memory consumption for the circuit in zkML for 7B-LLaMA model may be on the order of TB or even PB levels, making it totally impractical.
    \end{itemize}
    Limited by the long proof generation time and large memory consumption, zkML can only handle some small AI models with an extremely high service cost. While opML can provide proof in the native execution, and thus it is more cost-efficient and practical for some large models.
    \item \textbf{Security}: zkML can make ML onchain with the highest level of security possible, but is limited by proof generation time and other constraints. opML can make ML onchain with flexibility and performance, but may not be as secure as zkML because it is secured by economics, not by mathematics and cryptography like ZK. The security model of opML is similar to the current optimistic rollup systems. With the fraud proof system, opML can provide an any-trust guarantee: any one honest validator can force opML to behave correctly. Under the AnyTrust assumption, opML can maintain correctness and liveness.
    \item \textbf{Finality}: We can define the finalized point of zkML and opML as follows:
    \begin{itemize}
        \item zkML: zk proof of ML inference is generated (and verified)
        \item opML: challenge period of ML inference has ended
    \end{itemize}
    Though opML needs to wait for the challenge period, zkML also needs to wait for the generation of ZKP. When the proving time of zkML is longer than the challenge period of opML (it happens when ML model is large), opML can reach finality even faster than zkML.
\end{itemize}


The contribution of this work is summarized as follows:
\begin{itemize}
    \item We propose opML: optimistic machine learning on blockchain. 
    Compared to zkML, opML can provide ML service with low cost and high efficiency.
    \item We propose a novel multi-phase fraud proof protocol with semi-native execution and lazy loading design, which can greatly improve the efficiency of the fraud proof protocol and overcome the memory limitation in the fraud proof VM.
\end{itemize}

The organization of this paper goes as follows. In Section \ref{sec:arch}, we introduce the architecture and the workflow of opML. 
In Section \ref{sec:fpvm}, we introduce the fraud proof virtual machine (FPVM) in opML. In Section \ref{sec:ml_engine}, we design a high-efficient ML engine for both native execution and fruad proof. In Section \ref{sec:dispute}, we demonstrate the dispute game in opML.
In Section \ref{sec:multi}, we introduce the multi-phase dispute game in opML, which can help to overcome the limitations of the current fraud proof system. 
In Section \ref{sec:security}, we show that opML can achieve safety and liveness under the AnyTrust assumption.
In Section \ref{sec:incentive}, we investigate the incentive mechanism in opML to resolve the verifier dilemma.
In Section \ref{sec:dis}, we discuss the potential optimization and applications on opML. Section \ref{sec:relate} reviews the related work. Section \ref{sec:conclusion} concludes this paper with a final remark. 

\section{Overview}\label{sec:arch}
\begin{figure}
    \centering
    \includegraphics[width=0.9\textwidth]{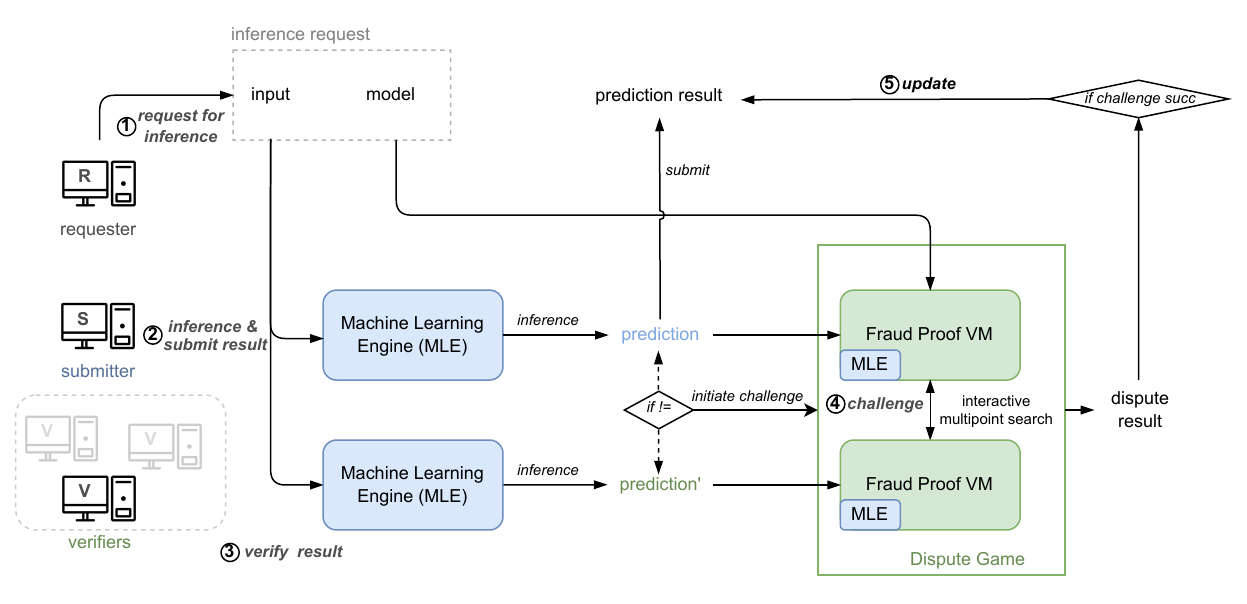}
    \caption{OPML Overview}
    \label{fig:overview}
\end{figure}
\subsection{Design Principles}

The design principles of opML are listed as follows:
\begin{itemize}
    \item \textit{Deterministic ML Execution}: The native ML execution may result in different execution results due to the randomness and computation on floating-point numbers. opML adopts the fixed-point arithmetic and the software-based floating point libraries to guarantee the consistency and determinism of the ML execution. Therefore, we can use a deterministic state transition function to represent the ML execution process.
    \item \textit{Separate Execution from Proving}: opML takes the same source code and compiles it twice. One is compiled for native execution, which is optimized for speed. And we allow using multi-thread CPU and GPU to speed up the native execution. The same code is also compiled to fraud proof VM instruction\footnote{In the current implementation, we adopts MIPS as the fraud proof VM. We are also exploring other efficient fraud proof VM such as WASM, RISC-V, etc.}, which will be used in the fraud proof protocol.
    This dual-target approach assures that execution is fast, while proving is based on the machine-independent code.
    \item \textit{Optimistic Machine Learning with Interactive Fraud Proofs}: opML adopts the interactive fraud proof, which will bisect a dispute down to a single instruction, and resolve the base-case using the on-chain fraud proof VM.
    \item \textit{Optimizing ML fraud proof with Multi-Phase Protocol}: The existing fraud proof systems that are widely adopted in the optimistic rollup systems need to cross-compile the whole computation into fraud proof VM instructions, which will result in inefficient execution and huge memory consumption. opML proposes a novel multi-phase protocol, which allows semi-native execution and lazy loading, which greatly speeds up the fraud proof process.
\end{itemize}

\subsection{Architecture}

opML relies on a fraud-proof mechanism to ensure the accuracy of machine learning results on the chain. The fraud-proof in opML comprises three integral components:
\begin{itemize}
    \item \textbf{Fraud Proof Virtual Machine (FPVM)}. Given a stateless program and its inputs, FPVM can trace any instruction step and prove it on layer 1 blockchain (L1).
    \item \textbf{Machine Learning Engine}. The highly efficient machine learning engine is designed to cater to both native execution and fraud-proof scenarios. This engine not only ensures a swift and accurate execution of machine learning tasks but also guarantees the consistency and determinism of the results. 
    \item \textbf{Interactive Dispute Game}. The dispute game will bisect a dispute down to a single instruction, and resolve the disputed one-step instruction using on-chain FPVM.
\end{itemize}

\subsection{Workflow}


In opML, the requester first initiates a task, and upon the server's completion, results are committed on-chain. Verifiers (Challengers) check these results, and if a dispute arises, the bisection protocol is initiated. This dispute resolution process aims to pinpoint and address any identified erroneous steps. Finally, smart contract arbitration concludes the workflow by decisively settling any disputes related to a single computation step. The opML workflow unfolds as follows:

\begin{enumerate}
    \item The requester first initiates an ML service task.
    \item The submitter then performs the ML service task and commits the result on chain.
    \item The verifier (challenger) will validate the results. If a verifier deems the results incorrect, he will initiate the dispute game (bisection protocol) with the server. The goal is to disprove the claim by pinpointing a specific erroneous step.
    \item The smart contract conducts arbitration concerning the disputed step, providing a conclusive resolution to the dispute.
    \item Finally, after a defined ``challenge period'', the results will be confirmed. 
\end{enumerate}

Both the server (submitter) and the verifier (challenger) need to stake in the system, and providing a incorrect result will lead to loss of the stake. Thus, if all parties follow their incentives, only the valid result will be committed.
Therefore, in the common case, the server (submitter) will always provide correct results, and the verifiers (challengers) will always find the results valid, and the dispute game would not happen. After a defined ``challenge period'', the results will be confirmed. 

\section{Fraud Proof Virtual Machine}\label{sec:fpvm}

\begin{figure}
    \centering
    \includegraphics[width=0.9\textwidth]{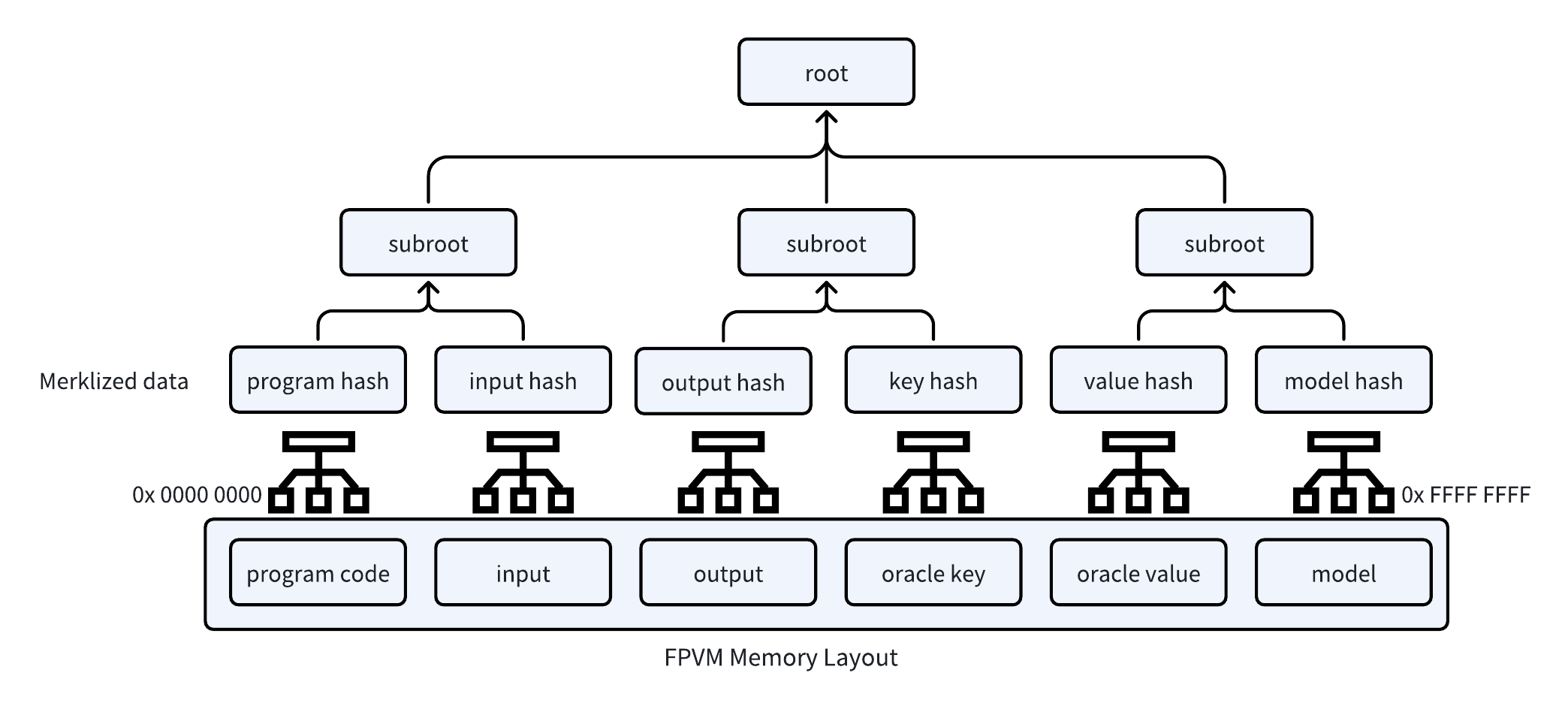}
    \caption{Fraud Proof VM Memory Layout}
    \label{fig:fpvm_memory}
\end{figure}

We build a fraud proof virtual machine (FPVM) for offchain execution and onchain arbitration. We guarantee the equivalence of the offchain VM and the onchain VM implemented on a smart contract.
Operationally, the FPVM is a state transition function. This state transition is referred to as a \textit{Step}, that executes a single instruction in the FPVM. Given an input state $S_{pre}$ steps on a single instruction encoded in the state, the FPVM, as a state transition function (STF) $VM(\cdot)$, will produce a new state $S_{post}$, i.e., $VM(S_{pre}) \rightarrow S_{post}$. 
Thus, the trace of a program executed by the FPVM is an ordered set of VM states, which highlights the effects of running a machine learning program on the FPVM. 
The execution trace $T$ is a sequence $(S_0,S_1,S_2, \cdots, S_n)$, where each $S_i$ is a VM state, and  $S_{i} = VM(S_{i-1}), \forall i \in [1,n]$. Each execution trace has an unique initial state $S_0$, which is determined by the ML model and input.

The VM state is managed with a Merkle tree \cite{mt}, only the Merkle root will be uploaded to the onchain smart contract. The Merkle root represents for the VM state.
As shown in Figure \ref{fig:fpvm_memory}, the memory in FPVM is divided into several areas, namely, ``program code'', ``input'', ``output'', ``oracle key'', ``oracle value'' and ``model''. The machine learning program is placed in the ``program'' field; and the input for the ML model is placed in the ``input'' field; the ML model is placed in the ``model'' field; the results of opML will be placed in the ``output'' field at the end of the execution of the ML program. There is a key-value oracle in fraud proof system, which is used by FPVM to read external data during its state transition. The oracle key and value will be placed in the ``oracle key'' field and ``oracle value'' field correspondingly.
Besides, the memory in FPVM is represented as a merkle tree. The tree has a fixed-depth of 27 levels, with leaf values of 32 bytes each. This spans the full 32-bit address space, where each leaf contains the memory at that part of the tree. The merkle root of the tree reflects the effects of memory writes in FPVM. With this merkle tree structure, each field in the FPVM memory can be represented as a merkle subtree root.

\section{Machine Learning Engine}\label{sec:ml_engine}

In opML, a highly efficient machine learning engine is designed to cater to both native execution and fraud-proof scenarios. This engine not only ensures a swift and accurate execution of machine learning tasks but also guarantees the consistency and determinism of the results. This becomes particularly crucial in the dispute resolution process, where the machine learning engine can statelessly provide a valid output and verify the dispute, reinforcing the reliability of the overall system.

\subsection{Separate Execution from Proving}

Following the design principle ``Separate Execution from Proving'', we design a high efficient machine learning engine for opML. We provide two kinds of implementation of the machein learning engine. One is compiled for native execution, which is optimized for high speed. And we use multi-thread CPU and GPU to speed up the native execution. Another is compiled to a fraud proof program for FPVM.
This dual-target approach assures that the execution is fast, while proving is based on the machine-independent code.

Consider the matrix multiplication in the machine learning engine, illustrated in Figure \ref{fig:ml:GPU}. Native execution employs GPU calculation (CUDA calculator \cite{sanders2010cuda}) for acceleration. As shown in Figure \ref{fig:ml:FPVM}, for the proving phase, the machine learning engine is compiled into machine-independent FPVM instructions, as CUDA calculator compatibility is lacking. The two implementations ensure consistent execution results.


\begin{figure}[!h]
	\begin{footnotesize}
		\begin{lstlisting}[language=C++]
// a,b,c are GPU device pointers to matrix
void gpu_matrix_mult(int *a,int *b, int *c, int m, int n, int k) { 
    int row = blockIdx.y * blockDim.y + threadIdx.y; 
    int col = blockIdx.x * blockDim.x + threadIdx.x;
    int sum = 0;
    for(int i = 0; i < n; i++) {
        sum += a[row * n + i] * b[i * k + col];
    }
    c[row * k + col] = sum;
} 
		\end{lstlisting}
	\end{footnotesize}
	\caption{Matrix multiplication in ML engine using GPU for native execution}
	\label{fig:ml:GPU}
\end{figure}

\begin{figure}[!h]
	\begin{footnotesize}
		\begin{lstlisting}[language=C++]
void fpvm_matrix_mult(int *a, int *b, int *c, int m, int n, int k) {
    for (int i = 0; i < m; ++i) {
        for (int j = 0; j < k; ++j) {
            int smu = 0;
            for (int h = 0; h < n; ++h) {
                smu += a[i * n + h] * b[h * k + j];
            }
            c[i * k + j] = smu;
        }
    }
}
		\end{lstlisting}
	\end{footnotesize}
	\caption{Matrix multiplication in ML engine to machine-independent code for FPVM}
	\label{fig:ml:FPVM}
\end{figure}

In the current implementation, to ensure the efficiency of AI model computation in the FPVM, we have implemented a lightweight machine learning engine specifically designed for this purpose instead of relying on popular ML frameworks like Tensorflow \cite{abadi2016tensorflow} or PyTorch \cite{paszke2019pytorch}. Additionally, a script that can convert Tensorflow and PyTorch models to this lightweight library is provided. With the provided script, we can train a model on popular ML frameworks like Tensorflow or PyTorch, and then convert it to opML model. 

\subsection{Consistency and Determinism}\label{sec:consistency}



The inconsistency in ML results can be attributed to two factors: randomness and variability in floating-point computations. To address randomness, fixing the random seed in the random number generator is a common practice since computer-generated randomness is essentially pseudo-randomness. Regarding the inconsistency in floating-point computations, during the native execution of DNN computations, especially across diverse hardware platforms, differences in execution results may arise due to the nature of floating-point numbers. 
For instance, parallel computations involving floating-point numbers, such as $(a + b) + c$ versus $a + (b + c)$, often yield non-identical outcomes due to rounding errors \cite{goldberg1991every}. Additionally, factors such as programming language, compiler version, and operating system can influence the computed results of floating-point numbers, leading to further inconsistency in ML results \cite{goldberg1991every}.

To tackle these challenges and guarantee the consistency of opML, we employ two key approaches:

\begin{itemize}
    \item Fixed-point arithmetic, also known as quantization technology \cite{yang2019quantization}, is adopted. This technique enables us to represent and perform computations using fixed precision rather than floating-point numbers. By doing so, we mitigate the effects of floating-point rounding errors, leading to more reliable and consistent results. It's important to note that using fixed-point arithmetic may result in a minimal loss of accuracy in DNN models. This tradeoff between execution performance and model accuracy is a key consideration when adopting such precision techniques.
    \item We leverage software-based floating-point libraries (softfloat) \cite{bsoft} that are designed to function consistently across different platforms. These libraries ensure cross-platform consistency and determinism of the ML results, regardless of the underlying hardware or software configurations.
\end{itemize}

By combining fixed-point arithmetic \cite{yang2019quantization} and software-based floating-point libraries (softfloat) \cite{bsoft}, we establish a robust foundation for achieving consistent and reliable ML results within the opML framework. 

\section{Interactive Dispute Game}\label{sec:dispute}


In the interactive dispute game, two or more parties (with at least one honest party) are assumed to execute the same program. Then, the parties can challenge each other with a pinpoint style to locate the disputed step. The step is sent to the smart contract on blockchain for arbitration.



\subsection{Locating the Disputed Point}

\begin{figure}
    \centering
    \includegraphics[width=0.8\textwidth]{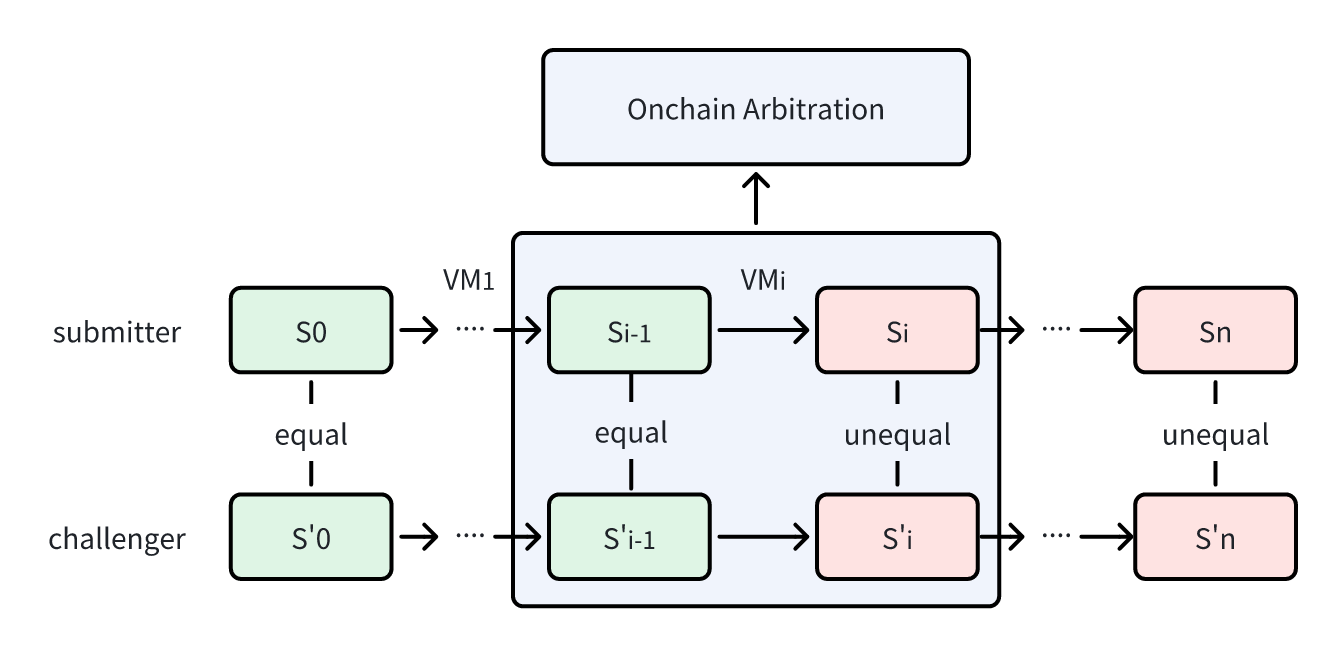}
    \caption{Dispute Game in opML}
    \label{fig:pinpoint}
\end{figure}

At the beginning, the submitter and the verifier agree on the initial state $S_0$, but disagree with the final state, i.e., $S_n \neq$ $S'_n$.
The objective of this protocol, as illustrated in Figure \ref{fig:pinpoint}, is to identify a specific VM$_k$ in a sequence of VM instructions executed within the context of a VM state $S_0$ (with $n$ such instructions in total), wherein $S_{k-1} =$ $S'_{k-1}$ but $S_k \neq$ $S'_k$.
 
The dispute game happens in a sequence of rounds. At the beginning of each round, the submitter and the verifier agree on a start VM state $S_i$ and they disagree on the end state $S_{i+j}$, for some $j>1$. Initially, $i=0$ and $j=n$. Then the challenger is now required to claim what state $S_m$ is at the midpoint VM state, where $m = i + \lfloor \frac{j}{2} \rfloor$. Next, the submitter must say whether he agrees or disagrees with the submitter's claimed midpoint state $S_m$. Now there are two cases:
\begin{itemize}
    \item If the submitter agrees with the challenger's midpoint state, the protocol has identified a smaller dispute: the challenger and the submitter agree on $S_m$, where $m = i + \lfloor \frac{j}{2} \rfloor$, but disagree with $S_{i+j}$. In the next round, we will let $i \leftarrow i + \lfloor \frac{j}{2} \rfloor$ and $j \leftarrow N - \lfloor \frac{j}{2} \rfloor$.
    \item If the submitter disagrees with the challenger's midpoint state, the protocol also has identified a smaller dispute: the challenger and the submitter agree on $S_i$ but disagree with $S_m$, where $m = i + \lfloor \frac{j}{2} \rfloor$. In the next round, we will let $i \leftarrow i$ and $j \leftarrow \lfloor \frac{j}{2} \rfloor$.
\end{itemize}
In both cases, the length of the dispute will be cut in half. The same procedure is repeated, cutting repeatly in half, until $j = 1$. 


The dispute game exhibits remarkable efficiency. In terms of time complexity, both the verifier and the submitter require only $\lceil \log_2 n \rceil$ rounds of challenge-response to converge on the same value of $k$, wherein $S_{k-1} =$ $S'_{k-1}$ but $S_k \neq$ $S'_k$. 
This synchronization between the submitter and the verifier is facilitated through a challenge-response mechanism featuring a timeout penalty \cite{kalodner2018arbitrum, truebit}.
The state $S_{k-1}$, along with auxiliary data, will be subsequently forwarded to a smart contract for arbitration.

\subsection{Onchain Arbitration}

For on-chain arbitration, the submitter and the verifier will send (VMI$_k$, S$_{k-1}$, S$_k$) to contract for arbitration. The onchain VM will take S$_{k-1}$ as the input and conduct one-step execution to output the correct S$_k$. 
Since the onchain VM only execute one instruction, the size of the data (witness) that needs to access is only $O(1)$.
For onchain arbitration, the witness consists of a partial expansion of the Merkle tree representing the before state $S_{k-1}$. The onchain VM uses the partially expended state tree to read the next instruction, emulate the instruction execution, and then compute the Merkle root hash of the resulting state. Note that the one-step onchain VM execution always takes only $O(1)$ computation and memory consumption, so that it is always possible to conduct the onchain arbitration using a feasible amount of Ethereum gas.


If the challenger produces a valid one-step proof in the onchain arbitration, he wins the challenge. Otherwise, the submitter wins the challenge.

\section{Multi-Phase Dispute Game}\label{sec:multi}

\subsection{Limitations of One-Phase Dispute Game}

With the design principle of ``Separate Execution from Proving'', we can achieve a high performance in the common optimistic case, where the submitter will always provide a correct result. However, to prevent the malicious behavior in the pessimistic case, we still need to generate a fraud proof for opML. In the fraud proof protocol of opML, we cross-compile the ML computation into fraud proof VM instructions and then start the dispute game to find the dispute step. However, cross-compiling the whole ML computation into fraud proof VM instructions has significant limitations:
 
\begin{itemize}
    \item \textbf{Low Execution Efficiency}: The one-phase dispute game has a critical drawback: for proving, all computations must be executed within the FPVM, preventing us from leveraging the full potential of GPU/TPU acceleration or parallel processing. Consequently, this restriction severely hampers the efficiency of providing fraud proof for large model inference, which also aligns with the current limitation of the referred delegation of the computation (RDoC) protocol \cite{canetti2011practical}.
    \item \textbf{Limited Memory in Fraud Proof VM}: The fraud proof VM has limited memory, for example, the MIPS VM can only support at most 4 GB memory. Due to the limited memory size, we can not load a large model into the fraud proof VM directly. For example, the size of a 7B-llama model in float64 is around 26GB, which can not be stored in the memory of the MIPS FPVM.
\end{itemize}

\subsection{Overview of Multi-Phase Protocol}

To address the constraints imposed by the one-phase protocol and ensure that opML can generate fraud proof in a performance level comparable to the native environment, we propose a multi-phase protocol. The multi-phase dispute game has the following properties that can well resolve the limitations above:

\begin{itemize}
    \item \textbf{Semi-Native Execution}: With the multi-phase design, we only need to conduct the computation in the VM only in the final phase, resembling the single-phase protocol. For other phases, we have the flexibility to perform computations that lead to state transitions in the native environment, leveraging the capabilities of parallel processing in CPU, GPU, or even TPU. By reducing the reliance on the VM, we significantly minimize overhead, resulting in a remarkable enhancement in the execution performance of opML, almost akin to that of the native environment.
    \item \textbf{Lazy Loading Design}: To optimize the memory usage and performance of the fraud proof VM, we implement a lazy loading technique. This means that we do not load all the data into the VM memory at once, but only the keys that identify each data item. When the VM needs to access a specific data item, it uses the key to fetch it from the external source and load it into the memory. Once the data item is no longer needed, it is swapped out of the memory to free up space for other data items. This way, we can handle large amounts of data without exceeding the memory capacity or compromising the efficiency of the VM.
\end{itemize}

\begin{figure}
    \centering
    \includegraphics[width=0.8\textwidth]{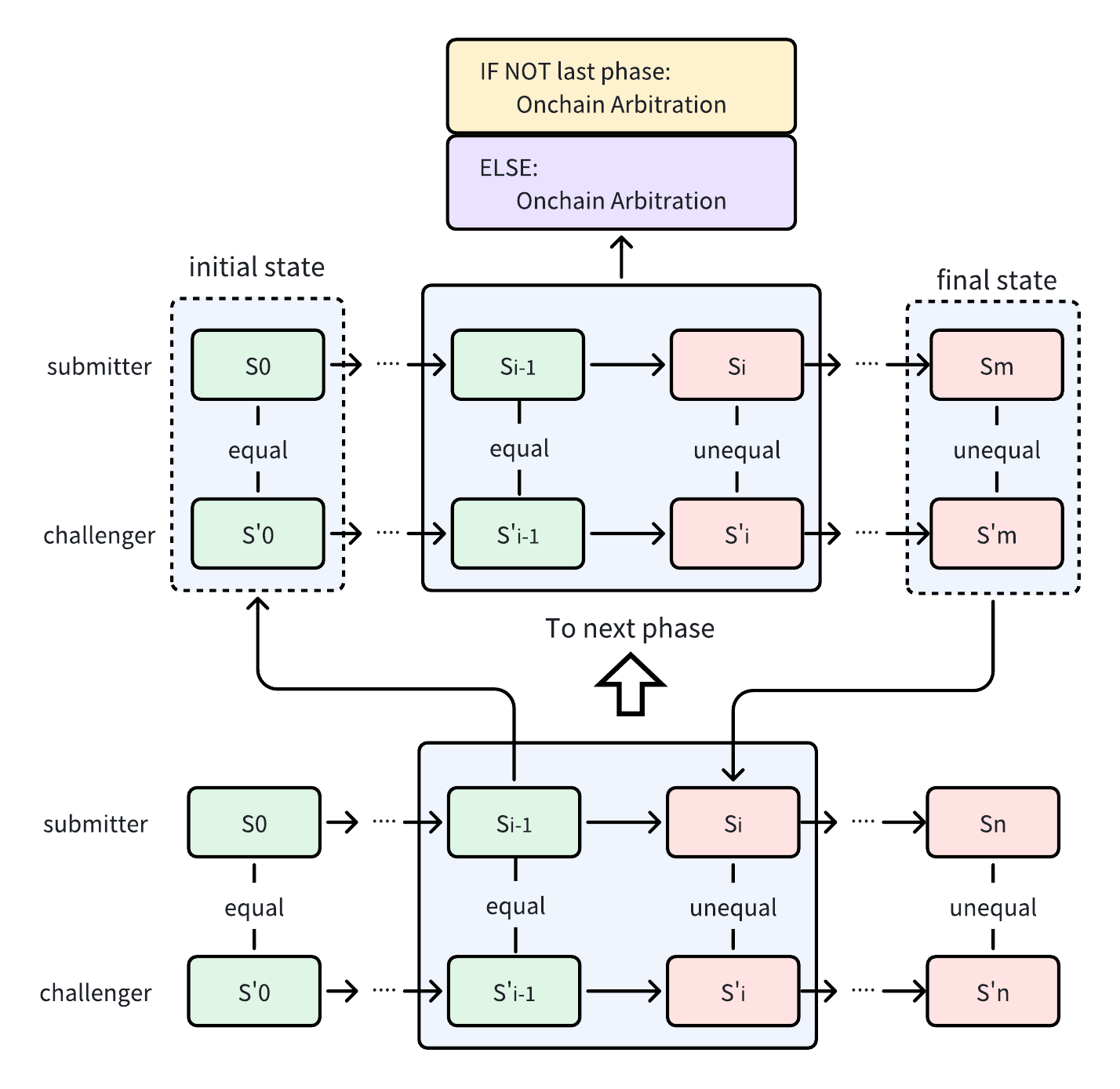}
    \caption{Multi-Phase Dispute Game}
    \label{fig:multi_phase}
\end{figure}

Figure \ref{fig:multi_phase} demonstrates a verification game consisting of two phases (k = 2). In Phase-2, the process resembles that of a single-phase verification game, where each state transition corresponds to a single VM micro-instruction that changes the VM state. In Phase-1, the state transition corresponds to a "Large Instruction" encompassing multiple micro-instructions that change the computation context.
The submitter and challenger will first start the dispute game on Phase-1 using a bisection protocol to locate the dispute step on a "large instruction". This step will be send to the next phase, Phase-2. Phase-2 works like the single-phase dispute game. The bisection protocol in Phase-2 will help to locate the dispute step on a VM micro-instruction. This step will be sent to the arbitration contract on the blockchain.

\subsection{State Transition to Next Phase}

To ensure the integrity and security of the transition to the next phase, we rely on the Merkle tree. As shown in Figure \ref{fig:multi_transition}, this operation involves extracting a Merkle sub-tree reconstruction, thus guaranteeing the seamless continuation of the dispute game process. Take the state transition in two-phase protocol as an example. On the phase-1, assume that the submitter and the challenger have already located the dispute step on a "large instruction", i.e., $S^k_{i-1} = S'^{k}_{i-1} \land S^k_{i} \neq S'^{k}_{i}$. Then we will build a Merkle tree on the state data $S^k_{i-1}$ to get the Merkle tree root $root(S^k_{i-1})$. We will construct the initial state in the next phase $S^{k+1}_0$ using the Merkle root. Specially, we need an empty VM image first, initialize it with the running program in the ``program code'' memory field in the FPVM, and then we will set the ``input'' memory field in FPVM as $root(S^k_{i-1})$, which will serves as a key for lazy loading. After completing initialization and filling in data, we can get the initial state in the next phase $S^k_{i-1} \Rightarrow S^{k+1}_0$. Due to the complexity of construction $S^{k+1}_0$ and the gas limitation of Ethereum, we build a zero-knowledge circuit to correctness of the construction of $S^{k+1}_0$ and verify the zero-knowledge proof on chain.
Similarly, at the end of the execution in phase 2, we need to check the consistency of the submitter's state $S^k_i \Leftrightarrow S^{k+1}_n$ and the challenger's state $S'^{k}_i \Leftrightarrow S'^{k+1}_n$. Specially, the state $S^k_i$ is stored as the ``output'' field of the FPVM memory, and we can prove it by providing a Merkle proof. 

\begin{figure}
    \centering
    \includegraphics[width=0.9\textwidth]{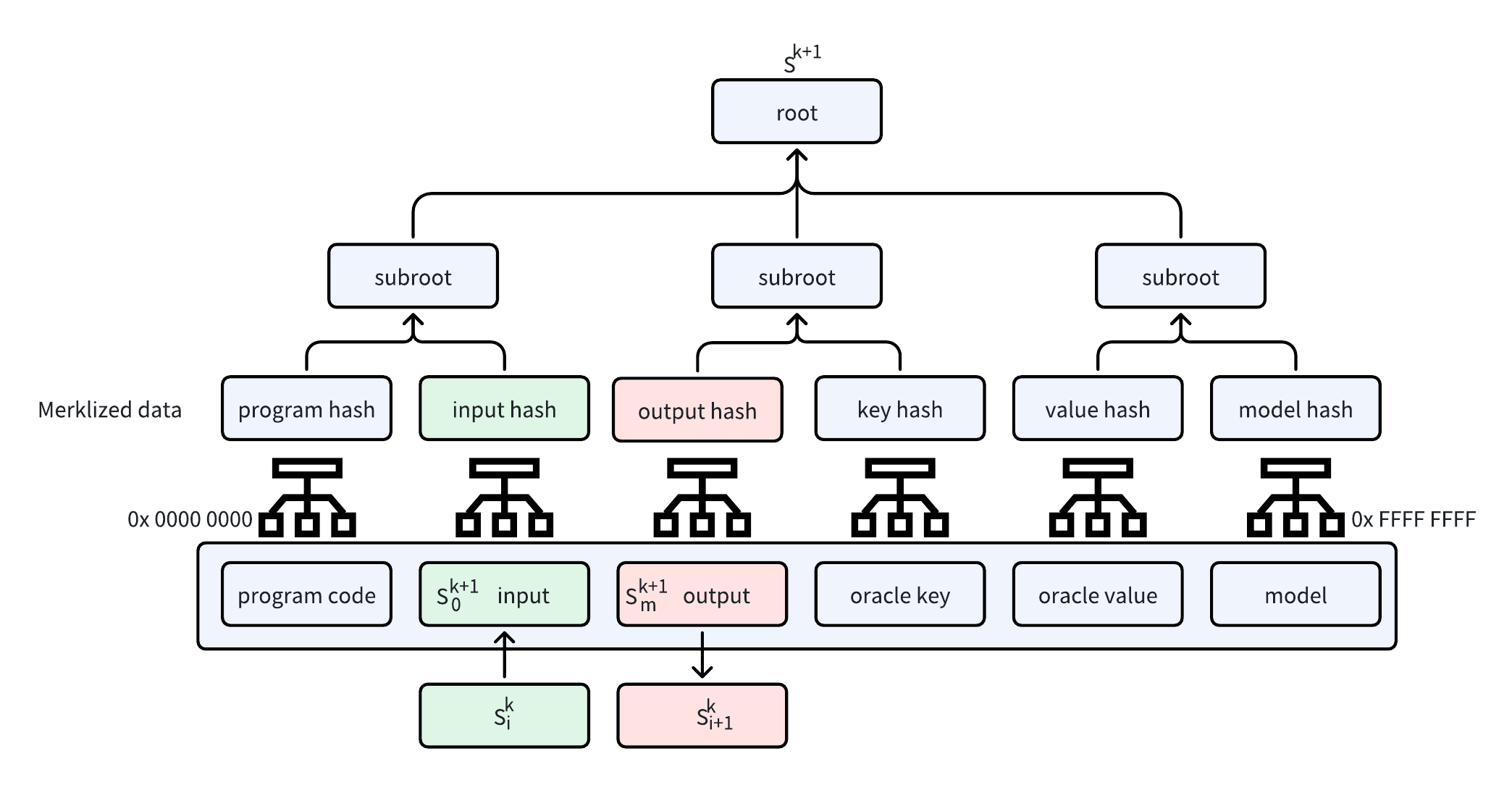}
    \caption{State Transition in Multi-Phase Dispute Game}
    \label{fig:multi_transition}
\end{figure}



\subsection{DNN Computation in Multi-Phase opML}

In this demonstration, we present a DNN computation in a two-phase opML approach:

\begin{itemize}
    \item The computation process of Machine Learning, specifically Deep Neural Networks (DNN), can be represented as a computation graph denoted as $G$. This graph consists of various computation nodes, capable of storing intermediate computation results.
    \item DNN model inference is essentially a computation process on the aforementioned computation graph. The entire graph can be considered as the inference state (computation context in Phase-1). As each node is computed, the results are stored within that node, thereby advancing the computation graph to its next state.
    \item Therefore, we can first conduct the dispute game on the computation graph (at Phase-1). On the phase-1 dispute game, the computation on nodes of the graph can be conducted in native environment using multi-thread CPU or GPU. The bisection protocol will help to locate the dispute node, and the computation of this node will be sent to the next phase (Phase-2) bisection protocol.
    \item In Phase-2 bisection, we transform the computation of a single node into Virtual Machine (VM) instructions, similar to what is done in the single-phase protocol.
\end{itemize}

It is worth noting that we anticipate introducing a multi-phase opML approach (comprising more than two phases) when the computation on a single node within the computation graph remains computationally complex. This extension will further enhance the overall efficiency and effectiveness of the fraud poof protocol.

\subsection{Performance Improvement}

In this section, we present a brief discussion and analysis of our proposed multi-phase fraud proof framework.

Suppose there are $n$ nodes in the DNN computation graph, and each node needs to take $m$ VM micro-instructions to complete the calculation in VM. Assuming that the speedup ratio on the calculation on each node using GPU or parallel computing is $\alpha$. This ratio represents the acceleration achieved through GPU or parallel computing and can reach significant values, often ranging from tens to even hundreds of times faster than the single-thread VM execution.
Based on these considerations, we draw the following conclusions:

\begin{itemize}
    \item \textbf{Performance}: Two-phase opML outperforms single-phase opML, achieving a computation speedup of $\alpha$ times. The utilization of multi-phase verification enables us to take advantage of the accelerated computation capabilities offered by GPU or parallel processing, leading to substantial gains in overall performance.
    \item \textbf{Space Complexity}: Two-phase opML reduces space complexity of the Merkle tree. When comparing the space complexity of the Merkle trees, we find that in two-phase opML, the size is $O(m + n)$, whereas in single-phase opML, the space complexity is significantly larger at $O(mn)$. The reduction in space complexity of the Merkle tree further highlights the efficiency and scalability of the multi-phase design.
\end{itemize}

In summary, the multi-phase fraud proof framework presents a remarkable performance improvement, ensuring more efficient and expedited computations, particularly when leveraging the speedup capabilities of GPU or parallel processing. Additionally, the reduced Merkle tree size adds to the system's effectiveness and scalability, making multi-phase opML a compelling choice for various applications.
\section{Security Analysis}\label{sec:security}

In this section, we provide the security analysis of our system. 
For simplicity, we use the AnyTrust assumption from Arbitrum \cite{kalodner2018arbitrum}.

\para{The AnyTrust assumption.} AnyTrust, rather than ``majority trust'', assumes that, for every claim, there is at least one honest node. Namely, either the submitter is honest, or at least one verifier is honest and will challenge within the pre-defined period. Suppose there are $m$ verifiers, $m-1$ of whom collude to stay silent about a submitter's wrong claim. Under AnyTrust, the only honest verifier will still succeed in disproving the wrong claim in front of the smart contract, and the wrong claim is thus rejected. Like the previous work \cite{kalodner2018arbitrum, truebit, zheng2021agatha}, we also assume data availability  
and anti-censorship (i.e., every verifier can always interact with the contract). They are solved by orthogonal countermeasures \cite{ipfs}.

\para{Safety and Liveness.}
We can show that opML can maintain safety and liveness under the AnyTrust assumption. 
\begin{itemize}
    \item \textbf{Safety}: \emph{The safety of opML under the AnyTrust assumption means that any one honest validator can force opML to behave correctly.} Suppose there exists only one honest node, and all other nodes are malicious. When a malicious node provides a incorrect result on chain, the honest node will always check the result and find it incorrect. At that time, the honest node can challenge the malicious one. The dispute game can force the honest node and the malicious node to pinpoint one concrete erroneous step. This step will be sent to the arbitration contract on the blockchain, and will find out the result provided by the malicious node is incorrect. Finally, the malicious node will be penalized and the honest node can get a corresponding reward. 
    \item \textbf{Liveness}: \emph{The Liveness of opML under the AnyTrust assumption means that any proposed result will be either accepted or rejected by the contract within a maximum period time.} Firstly, the instruction set and the execution traces of the opML program is finite, and can be done within $T_e$. Therefore, in the dispute game, the number of interactions is also finite. Therefore, the dispute game will end within a maximum period time $T_d$. If the provided result is correct, even in the worst case where $m-1$ malicious verifiers start the dispute game, the result will be accepted within $T_e + (m-1)\times T_d$. If the result is incorrect provided by a malicious verifier, the honest verifier will challenge it, even in the worst case where $m-2$ verifiers will start the dispute game to delay the actions of honest verifier, the result will be rejected within $T_e + (m-1)\times T_d$.
\end{itemize}

\para{Compared with ``AnyTrust'' and ``Majority Trust''}
In the following, we will conduct a brief security discussion between the system with ``AnyTrust'' and ``Majority Trust''.
As discussed above, opML is an ``AnyTrust'' system where any one honest validator can force opML to behave correctly. 
In contrast, systems like Layer 1 blockchains (e.g., Bitcoin \cite{nakamoto2008bitcoin} and Ethereum \cite{wood2014ethereum}) and Oracles (e.g., Chainlink \cite{breidenbach2021chainlink}) fall into the category of "Majority Trust" systems. These "Majority Trust" systems become susceptible to attacks when the number of malicious validators surpasses the threshold of $\lceil fm \rceil$, where $m$ represents the total number of validators, and $f$ denotes the Byzantine fault tolerance ratio. The specific value of $f$ varies with different consensus algorithms; for instance, it is 0.5 for Proof of Work (PoW) blockchains like Bitcoin. 
Consider a system with $m$ validators, and let $p$ denote the probability of a validator being malicious.
The probability that a ``Majority Trust'' system  is secure is 
\begin{equation}\nonumber
    P_{\text{majority-trust}} = \sum_{i=0}^{\lceil fm \rceil} {m \choose i} p^i(1-p)^{m-i} = ( m-k) {m \choose k} \int_0^{1-p} t^{n-k-1}(1-t)^k dt.
\end{equation}
The probability that a system with ``AnyTrust'' is secure is
\begin{equation}\nonumber
    P_{\text{any-trust}} = \sum_{i=0}^{m-1} {m \choose i} p^i(1-p)^{m-i} = 1 - p^m.
\end{equation}
It is evident that the "AnyTrust" system exhibits superior security compared to the "Majority Trust" system, specifically denoted as $P_{\text{any-trust}} > P_{\text{majority-trust}}$. Additionally, as illustrated in Figure \ref{fig:secure-probability}, the "AnyTrust" system attains a high level of security even with a limited number of nodes. For instance, even with a relatively high probability of 0.5 for a validator to be malicious (indicating that half of the validators can potentially be malicious), the "AnyTrust" system maintains a 99.9\% security probability with only 10 validators. In contrast, the "Majority Trust" system struggles to provide such a robust security guarantee under similar conditions.
\begin{figure*}
    \centering
    \subfigure[Configuration $p=0.4,f=0.5$]{
    \includegraphics[width=0.45\linewidth]{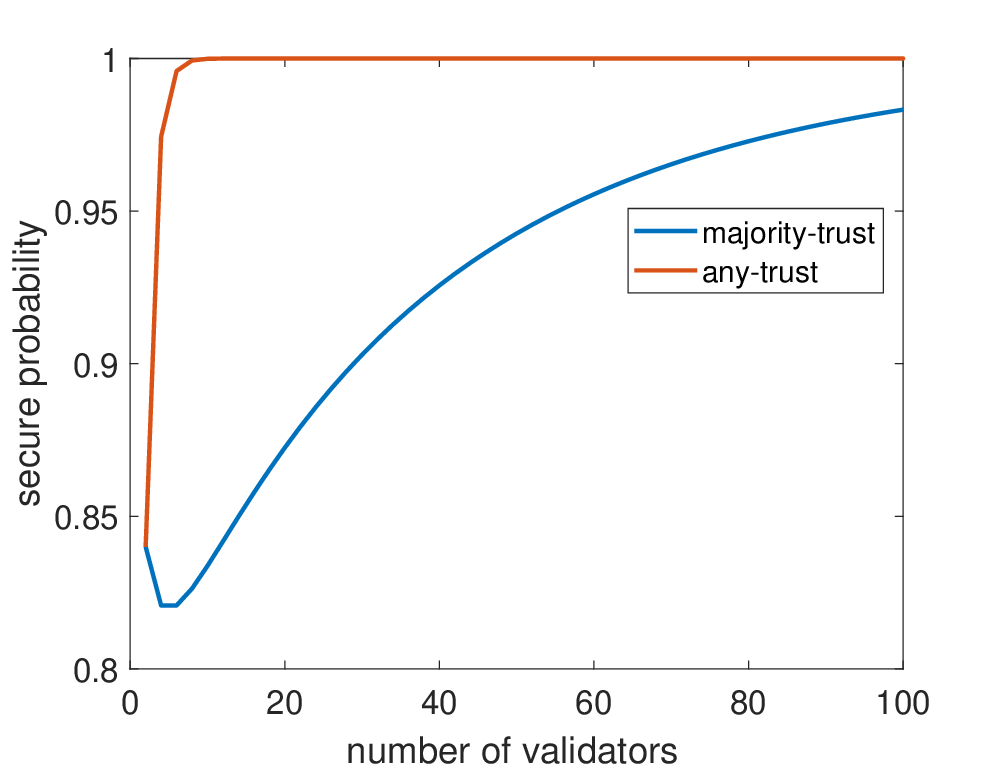}
    }
    \subfigure[Configuration $p=0.5,f=0.5$]{
    \includegraphics[width=0.45\linewidth]{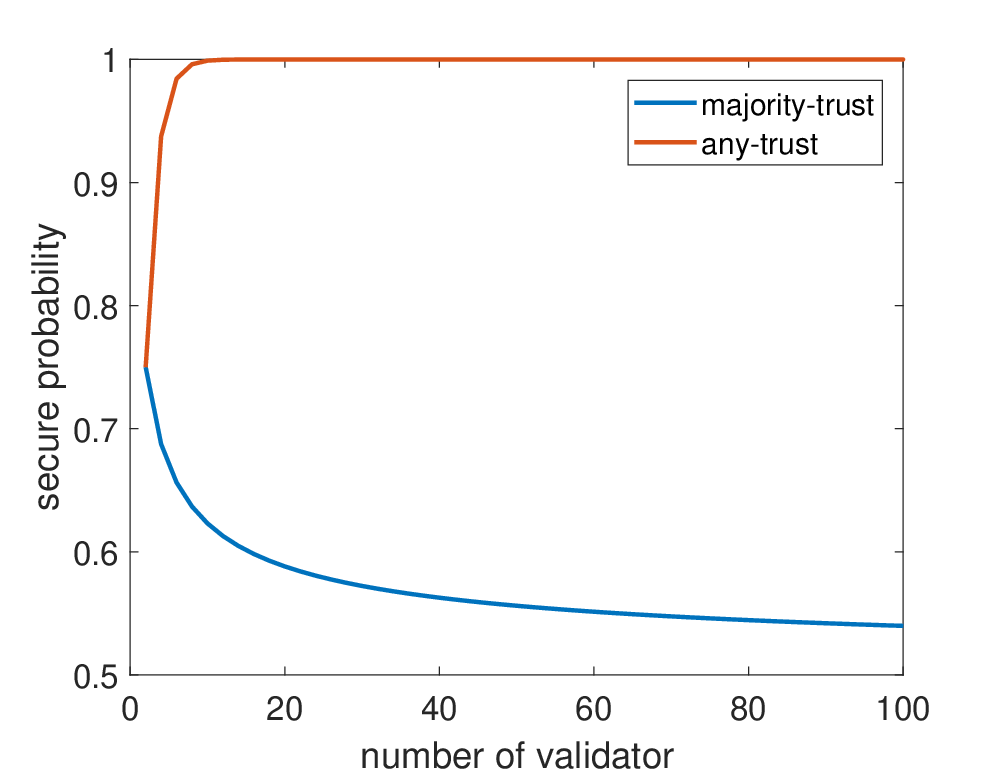}
    }
    \caption{Secure probability of ``Majority Trust'' and ``AnyTrust'' systems}
    \label{fig:secure-probability}
\end{figure*}


\section{Incentive Mechanism}\label{sec:incentive}

To achieve the safety and liveness of opML, we need to guarantee that at least one honest validator will check the results. Therefore, it is important to design an incentive-compatible mechanism where the rational validators will always check the results and the rational submitter will never cheat.

\subsection{Verifier Dilemma}

Like the optimistic rollup systems, opML can also suffer from the \textit{Verifier's Dilemma} \cite{luu2015demystifying}.
Here we consider a simple verification game $\mathbb{G}$.
Assume that all the validators (including the submitter) need to stake $S$ on chain, and the computation cost for opML execution is $C$. If the submitter provides an incorrect and nobody challenges him, he can benefit $B$ from cheating, but all the validator will suffer from a loss $L$. If the submitter provides an incorrect and it is successfully challenged by one validator, then the submitter will lose all his stake, and the validator can get the corresponding reward $R$. The action for the submitter is to cheat or not. And the action for the validator is to check or not. The payoff matrix is shown in Table \ref{table:verifier}, and the Nash equilibrium strategy \cite{roughgarden2010algorithmic} in this verification game is shown in the following theorem.

\begin{table}[h]
\begin{tabular}{|c|c|c|}
\hline
\textbf{validator/submitter} & \textbf{cheat} & \textbf{No cheat} \\ \hline
\textbf{validate}            & $(R - C, -S)$    & $(-C, -C)$          \\ \hline
\textbf{No validate}         & $(-L, B)$        & $(0,-C) $           \\ \hline
\end{tabular}
\caption{Payoff matrix in the verification game}
\label{table:verifier}
\end{table}

\begin{theorem}
    The mixed strategy Nash equilibrium of this verification game is that the submitter will randomly cheat with probability $p_c = \frac{C}{R+L}$, and the validator will randomly check with probability $p_v = \frac{B+C}{B+S}$.
\end{theorem}

\begin{proof}
    According to the payoff matrix in Table \ref{table:verifier}, we can find that there is no dominant strategy for both validator and submitter. Therefore, there is no pure strategy Nash equilibrium in the verification game. We now turn to consider the mixed strategy one. Assume that the submitter will randomly cheat with probability $p_c$, and the validator will randomly check with probability $p_v$. Then for the mixed strategy Nash equilibrium, we have
    \begin{equation}\nonumber
    \begin{cases}
        p_c (R - C) + (1 - p_c) (-C) = - p_c L  & \\
        p_v (-S) + (1-p_v) B = -C & \\
    \end{cases}
    \end{equation}
    By solving the equations above, we have that $p_c = \frac{C}{R+L}$ and $p_v = \frac{B+C}{B+S}$.
    The proof is thus completed.
\end{proof}

There is no pure strategy Nash equilibrium in this verification game. This means that a rational submitter can randomly cheat, while the rational validator will also randomly check. This is not what we want. What's worse is if the submitter cheat randomly, with probability of cheating less than $c/(R+L)$, then a rational validator will never check, so the submitter will never get caught cheating. This is indeed the verifier dilemma, where there's no point in checking somebody's work if you know they're not going to cheat; but if you're not going to check then they have an incentive to cheat.

\subsection{Attention Challenge}

To solve the verifier dillema, we introduce the \textit{Attention Challenge} \cite{mamageishvili2023incentive} into opML.

In opML, we denote $x$ as the user's input to opML (corresponding to the initial state $S_0$), $f(\cdot)$ as the deterministic computation function, and $f(x)$ is the opML result (corresponding to the final state $S_n$). Instead of revealing $f(x)$ on chain, the submitter will first reveal $H(A_s, f(x))$ on chain, where $A_s$ is the address of the submitter and $H(\cdot)$ is the hash function. 

There exists a window of time, where the validator whose hash value $H(A_v, f(x)) < T$ should respond on chain, where $A_v$ is the address of the validator, and $T$ is a chosen threshold. After that, the submitter can post the opML result $f(x)$ on chain, which can be challenged as normal if any validator disagrees with it. The submitter can accuse any validator who responded incorrectly or did not respond. This accusation will be checked on chain and the validator who responded incorrectly or did not respond will be penalized by $G$. Note that the accusation will be valid only after the submitter's claim $f(x)$ is eventually accepted as correct. If any validator is penalized, the submitter will get half of the penalty $\frac{1}{2}G$, and the other half will be burned.

We denote $p_t = \Pr(H(A_v, f(x)) < T)$ as the probability that the validator needs to respond, which can be adjusted by chosen a proper $T$.
Then if the validator checks the results, then his utility is $p_c R - C$, where $p_c$ is the probability that the submitter checks and $C$ is the computation cost. If the validator does not check the results, he may be punished by no-responding when he needs to respond. In this case, the validator's utility is $-p_c \cdot L - p_t\cdot G$. In the following theorem, we can show that when $p_t \cdot G > C$, the rational validator will always check, and the rational submitter will never cheat. 

\begin{theorem}
    When $p_t \cdot G > C$, the only Nash equilibrium of the verification game with attention challenge mechanism is that validator will always check and the submitter will never cheat.
\end{theorem}

\begin{proof}
    The utility for the validator to check the results is $U(\text{check}) = p_c R - C$, and the utility for the validator to be lazy and not check the results is $U(\text{lazy}) = -p_c \cdot L - p_t\cdot G$. When $p_t \cdot G > C$, we have that
    \begin{equation}\nonumber
        \begin{aligned}
            U(\text{check}) - U(\text{lazy}) = p_c(R+C) + p_tG-C >0, \quad \forall p_c \in [0,1]
        \end{aligned}
    \end{equation}
    Therefore, no matter what probability the submitter cheats, the dominant strategy for the validator is to always check. 
    Given that the validator will always check, the dominant strategy for the submitter is to always provide a correct result.
    The proof is thus completed.
\end{proof}


Besides, the cost of introducing the attention challenge mechanism in opML is very low. Assume that the interest rate of locking up the stake $G$ is $r$, and the gas fee for the validator to respond is $t$. Then the extra cost for the attention challenge mechanism in opML is $rG + tp_t$. The minimal cost is $rG + tp_t \geq 2\sqrt{rtC}$ when $G = \sqrt{\frac{tC}{r}}$ and $p_t = \sqrt{\frac{rC}{t}}$.

Specifically, as shown in \cite{a16z}, A GPT-3.5 inference takes approximately 1 second on an A100 with the raw computation cost \$0.001 for 1,000 tokens (compared to OpenAI’s pricing of \$0.002/1000 tokens \cite{openai}), i.e., $C = \$0.001$. With the interest rate $r=0.001$ and the gas fee $t=\$1$, we can set the probability $p_t = 0.1\%$ and the validator needs to deposit $G=1$. The total cost for the attention challenge per opML request is only $\$0.002$. 
\section{Discussion}\label{sec:dis}

\subsection{Optimization with Multi-Sect Dispute Points and Multi-Step zkVM}

Furthermore, we can utilize multiple checkpoints and multi-step ZK verifier to provide a zk-fraud proof to speed up the finality of opML. Assuming that there are $N$ steps in the opML fraud proof VM. Then the number of interactions would be $O(\log_2(N))$. One optimization approach to reduce to the interactions is to use a multi-step on-chain arbitrator, instead of only one-step on-chain arbitrator.

This can be done by using a zkVM, e.g., zkWASM \cite{gao2022zawa}, RICS Zero \cite{risv0}, Cairo VM \cite{goldberg2021cairo}, which can provide zk-proof for the multi-step execution. Suppose the zkVM can generate a zk-proof for $m$ steps, then the number of interactions can be reduce from $O(\log_2(N))$ to $O(\log_2(N/m))$. Besides, instead of using bisection protocol with only 1 checkpoint, we can provide multiple checkpoints (for example $k$ ($k >1$) checkpoints) in the dispute game (verification game). At this time, we can further reduce the interactions number from $O (\log_2(N/m))$ to $O( \log_{k+1}(N/m))$. With these optimizations, we can greatly speed up the finality of opML, and it is possible to achieve almost instant confirmation!


\subsection{Training/Fine-tuning in opML}

In opML, the use of a fraud-proof VM plays a pivotal role in ensuring the correctness of ML results. This specialized VM has the capability to support general computations, encompassing both inference and training tasks. The training and fine-tuning process can also be conceptualized as a series of state transitions when model parameters undergo deterministic updates.

Specifically, the inference phase of a DNN model involves straightforward forward computation on the DNN computation graph, denoted as $G$. On the other hand, the training process encompasses both forward computation and backward update (backpropagation) on the same DNN computation graph $G$. While the tasks differ in their objectives, the computation processes for forward computation and backward update are similar, enabling a unified approach.

Through the integration of a multi-phase opML approach, we can extend support to the training process efficiently as well. Here’s how it works: The dispute game initiates during each iteration of the training process, to locate a dispute within a specific iteration. Subsequently, the process advances to the next phase, where the submitter and challenger engage in a dispute protocol on the computation graph for both forward and backward processes. This allows for the pinpointing of the dispute node, and the computation related to this node is then forwarded to the subsequent phase, where the fraud-proof VM arbitrates the issue.

opML's extension to the training process has significant implications for verifying ML model generation on the blockchain. By utilizing on-chain data to train and update the ML model, opML ensures that the process is auditable and transparent. Moreover, opML's integration with the training process provides an effective means to validate ML model updates, safeguarding against potential backdoors \cite{hong2022handcrafted} and ensuring the model's integrity and security.

\subsection{Combining zkML and opML for Enhanced Privacy}

In contrast to zkML, opML lacks intrinsic privacy features. However, a compelling approach to augment privacy is by integrating the strengths of both opML and zkML. One feasible strategy involves leveraging zkML for input processing in the early layers of the model and subsequently employing opML for the remaining layers. This amalgamation allows us to attain a nuanced form of privacy—while not achieving full zk-level privacy, it is possible to ensure a secure and non-reversible level of obfuscation. A crucial consideration in this integration is determining the optimal number of layers to implement through zkML, striking a delicate balance between computational efficiency and heightened security.

\section{Related Work}\label{sec:relate}

\textbf{AI computation on blockchain.} 
The concept of AI computation on the blockchain has garnered attention from both academia and industry. In the academic realm, Das et al. have suggested the potential for "Computationally Intensive Smart Contracts" that encompass machine learning tasks \cite{yoda}. Teutsch et al. have explored the application of "Autonomous machine learning" in the context of TrueBit \cite{truebit}, while Wüst et al. have recognized the current limitations of smart contracts for machine learning applications \cite{ace}. 
Within the industry, Microsoft has put forth the concept of "Decentralized and Collaborative AI" on blockchain platforms \cite{microsoft}. 
Several startups, including Cortex \cite{cortex}, SingularityNET \cite{singularitynet}, Algorithmia \cite{algorithmia} have also embraced this vision.
However, the existing proposals typically demonstrate only limited AI computations, achieving consensus on a small number of nodes \cite{microsoft,algorithmia,singularitynet,cortex} or provide insufficient technical details \cite{yoda,truebit,ace}. These approaches lack scalability, rendering them impractical for real-world deep neural network computations within Ethereum smart contracts.

\textbf{Zero-Knowledge Machine Learning (zkML).} 
In \cite{liu2021zkcnn}, the authors provide zero-knowledge proof for an CNN model. In \cite{zhang2020zero}, the authors develop approaches to efficiently turn decision tree predictions and accuracy into statements of zero knowledge proofs. In \cite{wang2022ezdps}, the authors propose ezDPS: an efficient and zero-knowledge machine learning inference pipeline. In \cite{xing2023zero}, the authors propose a Zero-Knowledge Proof-based Federated Learning (ZKP-FL) scheme on blockchain.
In \cite{weng2021mystique}, the authors propose Mystique, an efficient conversions for zero-knowledge proofs with applications to machine learning, which can improve the execution time by a factor by 7$\times$.
In \cite{garg2023experimenting}, the authors propose zero-knowledge proof of training (zkPoT). However, due to its huge cost, zkPoT only works for the logistic regression, and can not be applied to some larger models. 


Our work also resembles and references the work on Agatha \cite{zheng2021agatha}. Agatha system is developed by Microsoft Research, which demonstrates the practical contract verification for DNN computation on public Ethereum. Agatha adopts the fraud proof to verify the computation results. Instead of using fraud proof VM, Agatha adopts the graph-based pinpoint protocol (GPP) which enables the pinpoint protocol on computational graphs.
Since Agatha is not open-sourced, we can not compare the performance of Agatha and opML directly.
Generally, our work distinguishes from Agatha \cite{zheng2021agatha} and the existing works in the following perspectives:
\begin{itemize}
    \item Agatha does not support Turing-complete computation. Agatha enables the native DNN computation by the graph-based pinpoint protocol (GPP) instead of VM-based pinpoint protocol.  
    opML is designed to enable Turing-complete computation. 
    \item With our Turing-complete fraud-proof VM, opML can handle various aspects of ML, including model inputs and outputs process, and can support various ML models. This versatility allows opML to support ML training and fine-tuning, a feature that Agatha lacks. For applications like LLaMA that require data processing, multiple inferences, and intermediate result storage, going beyond typical DNN model calculation, where a Turing-complete VM is essential. opML excels in this context, outperforming Agatha.
    \item With our Turing-complete fraud-proof VM, opML extends its capabilities beyond deep neural network (DNN) computation. It accommodates various machine learning algorithms, including K-nearest neighbors (KNN) \cite{soucy2001simple}, decision trees \cite{song2015decision}, random forests \cite{cutler2012random}, and more, enhancing its versatility and potential for a broader range of AI applications.
\end{itemize} 

\section{Conclusions}\label{sec:conclusion}

In conclusion, the optimistic approach to on-chain AI and machine learning brings substantial advantages over existing methods, enabling the fusion of AI power with the integrity and security of the blockchain. opML excels by delivering a cost-effective and efficient ML service compared to zkML. This positions opML as a key player in reshaping the future of decentralized, secure, and transparent AI services. As the landscape of on-chain AI evolves, opML stands as a pivotal solution, unlocking the full potential of these technologies and ensuring a transformative journey toward accessible, secure, and efficient on-chain machine learning.



\bibliographystyle{ACM-Reference-Format}
\bibliography{reference,zkml}


\begin{thebibliography}{50}


\ifx \showCODEN    \undefined \def \showCODEN     #1{\unskip}     \fi
\ifx \showDOI      \undefined \def \showDOI       #1{#1}\fi
\ifx \showISBNx    \undefined \def \showISBNx     #1{\unskip}     \fi
\ifx \showISBNxiii \undefined \def \showISBNxiii  #1{\unskip}     \fi
\ifx \showISSN     \undefined \def \showISSN      #1{\unskip}     \fi
\ifx \showLCCN     \undefined \def \showLCCN      #1{\unskip}     \fi
\ifx \shownote     \undefined \def \shownote      #1{#1}          \fi
\ifx \showarticletitle \undefined \def \showarticletitle #1{#1}   \fi
\ifx \showURL      \undefined \def \showURL       {\relax}        \fi
\providecommand\bibfield[2]{#2}
\providecommand\bibinfo[2]{#2}
\providecommand\natexlab[1]{#1}
\providecommand\showeprint[2][]{arXiv:#2}

\bibitem[Abadi et~al\mbox{.}(2016)]%
        {abadi2016tensorflow}
\bibfield{author}{\bibinfo{person}{Mart{\'\i}n Abadi}, \bibinfo{person}{Paul Barham}, \bibinfo{person}{Jianmin Chen}, \bibinfo{person}{Zhifeng Chen}, \bibinfo{person}{Andy Davis}, \bibinfo{person}{Jeffrey Dean}, \bibinfo{person}{Matthieu Devin}, \bibinfo{person}{Sanjay Ghemawat}, \bibinfo{person}{Geoffrey Irving}, \bibinfo{person}{Michael Isard}, {et~al\mbox{.}}} \bibinfo{year}{2016}\natexlab{}.
\newblock \showarticletitle{Tensorflow: a system for large-scale machine learning.}. In \bibinfo{booktitle}{\emph{Osdi}}, Vol.~\bibinfo{volume}{16}. Savannah, GA, USA, \bibinfo{pages}{265--283}.
\newblock


\bibitem[Ben-Sasson et~al\mbox{.}(2019)]%
        {ben2019scalable}
\bibfield{author}{\bibinfo{person}{Eli Ben-Sasson}, \bibinfo{person}{Iddo Bentov}, \bibinfo{person}{Yinon Horesh}, {and} \bibinfo{person}{Michael Riabzev}.} \bibinfo{year}{2019}\natexlab{}.
\newblock \showarticletitle{Scalable zero knowledge with no trusted setup}. In \bibinfo{booktitle}{\emph{Advances in Cryptology--CRYPTO 2019: 39th Annual International Cryptology Conference, Santa Barbara, CA, USA, August 18--22, 2019, Proceedings, Part III 39}}. Springer, \bibinfo{pages}{701--732}.
\newblock


\bibitem[Ben-Sasson et~al\mbox{.}(2014)]%
        {ben2014succinct}
\bibfield{author}{\bibinfo{person}{Eli Ben-Sasson}, \bibinfo{person}{Alessandro Chiesa}, \bibinfo{person}{Eran Tromer}, {and} \bibinfo{person}{Madars Virza}.} \bibinfo{year}{2014}\natexlab{}.
\newblock \showarticletitle{Succinct $\{$Non-Interactive$\}$ zero knowledge for a von neumann architecture}. In \bibinfo{booktitle}{\emph{23rd USENIX Security Symposium (USENIX Security 14)}}. \bibinfo{pages}{781--796}.
\newblock


\bibitem[Benet(2014)]%
        {ipfs}
\bibfield{author}{\bibinfo{person}{Juan Benet}.} \bibinfo{year}{2014}\natexlab{}.
\newblock \showarticletitle{Ipfs-content addressed, versioned, p2p file system}.
\newblock \bibinfo{journal}{\emph{arXiv preprint arXiv:1407.3561}} (\bibinfo{year}{2014}).
\newblock


\bibitem[Bootle et~al\mbox{.}(2016)]%
        {bootle2016efficient}
\bibfield{author}{\bibinfo{person}{Jonathan Bootle}, \bibinfo{person}{Andrea Cerulli}, \bibinfo{person}{Pyrros Chaidos}, \bibinfo{person}{Jens Groth}, {and} \bibinfo{person}{Christophe Petit}.} \bibinfo{year}{2016}\natexlab{}.
\newblock \showarticletitle{Efficient zero-knowledge arguments for arithmetic circuits in the discrete log setting}. In \bibinfo{booktitle}{\emph{Advances in Cryptology--EUROCRYPT 2016: 35th Annual International Conference on the Theory and Applications of Cryptographic Techniques, Vienna, Austria, May 8-12, 2016, Proceedings, Part II 35}}. Springer, \bibinfo{pages}{327--357}.
\newblock


\bibitem[Breidenbach et~al\mbox{.}(2021)]%
        {breidenbach2021chainlink}
\bibfield{author}{\bibinfo{person}{Lorenz Breidenbach}, \bibinfo{person}{Christian Cachin}, \bibinfo{person}{Benedict Chan}, \bibinfo{person}{Alex Coventry}, \bibinfo{person}{Steve Ellis}, \bibinfo{person}{Ari Juels}, \bibinfo{person}{Farinaz Koushanfar}, \bibinfo{person}{Andrew Miller}, \bibinfo{person}{Brendan Magauran}, \bibinfo{person}{Daniel Moroz}, {et~al\mbox{.}}} \bibinfo{year}{2021}\natexlab{}.
\newblock \showarticletitle{Chainlink 2.0: Next steps in the evolution of decentralized oracle networks}.
\newblock \bibinfo{journal}{\emph{Chainlink Labs}}  \bibinfo{volume}{1} (\bibinfo{year}{2021}), \bibinfo{pages}{1--136}.
\newblock


\bibitem[Canetti et~al\mbox{.}(2011)]%
        {canetti2011practical}
\bibfield{author}{\bibinfo{person}{Ran Canetti}, \bibinfo{person}{Ben Riva}, {and} \bibinfo{person}{Guy~N Rothblum}.} \bibinfo{year}{2011}\natexlab{}.
\newblock \showarticletitle{Practical delegation of computation using multiple servers}. In \bibinfo{booktitle}{\emph{Proceedings of the 18th ACM conference on Computer and communications security}}. \bibinfo{pages}{445--454}.
\newblock


\bibitem[Chen et~al\mbox{.}(2018)]%
        {cortex}
\bibfield{author}{\bibinfo{person}{Z Chen}, \bibinfo{person}{W Wang}, \bibinfo{person}{X Yan}, {and} \bibinfo{person}{J Tian}.} \bibinfo{year}{2018}\natexlab{}.
\newblock \showarticletitle{Cortex-AI on blockchain: The decentralized AI autonomous system}.
\newblock \bibinfo{journal}{\emph{Cortex White Paper}} (\bibinfo{year}{2018}).
\newblock


\bibitem[Cutler et~al\mbox{.}(2012)]%
        {cutler2012random}
\bibfield{author}{\bibinfo{person}{Adele Cutler}, \bibinfo{person}{D~Richard Cutler}, {and} \bibinfo{person}{John~R Stevens}.} \bibinfo{year}{2012}\natexlab{}.
\newblock \showarticletitle{Random forests}.
\newblock \bibinfo{journal}{\emph{Ensemble machine learning: Methods and applications}} (\bibinfo{year}{2012}), \bibinfo{pages}{157--175}.
\newblock


\bibitem[Das et~al\mbox{.}(2018)]%
        {yoda}
\bibfield{author}{\bibinfo{person}{Sourav Das}, \bibinfo{person}{Vinay~Joseph Ribeiro}, {and} \bibinfo{person}{Abhijeet Anand}.} \bibinfo{year}{2018}\natexlab{}.
\newblock \showarticletitle{Yoda: Enabling computationally intensive contracts on blockchains with byzantine and selfish nodes}.
\newblock \bibinfo{journal}{\emph{arXiv preprint arXiv:1811.03265}} (\bibinfo{year}{2018}).
\newblock


\bibitem[Dinh and Thai(2018)]%
        {dinh2018ai}
\bibfield{author}{\bibinfo{person}{Thang~N Dinh} {and} \bibinfo{person}{My~T Thai}.} \bibinfo{year}{2018}\natexlab{}.
\newblock \showarticletitle{AI and blockchain: A disruptive integration}.
\newblock \bibinfo{journal}{\emph{Computer}} \bibinfo{volume}{51}, \bibinfo{number}{9} (\bibinfo{year}{2018}), \bibinfo{pages}{48--53}.
\newblock


\bibitem[ezkl({[n.\,d.]})]%
        {nanoGPT}
\bibfield{author}{\bibinfo{person}{ezkl}.} \bibinfo{year}{[n.\,d.]}\natexlab{}.
\newblock \bibinfo{title}{zkSnark nanoGPT}.
\newblock
\newblock
\newblock
\shownote{\url{https://hackmd.io/mGwARMgvSeq2nGvQWLL2Ww##Honey-I-SNARKED-the-GPT}}.


\bibitem[Gao et~al\mbox{.}(2022)]%
        {gao2022zawa}
\bibfield{author}{\bibinfo{person}{Sinka Gao}, \bibinfo{person}{Guoqiang Li}, \bibinfo{person}{Hongfei Fu}, \bibinfo{person}{Heng Zhang}, {and} \bibinfo{person}{Junyu Zhang}.} \bibinfo{year}{2022}\natexlab{}.
\newblock \showarticletitle{ZAWA: A ZKSNARK WASM Emulator}.
\newblock \bibinfo{journal}{\emph{Proc. ACM Program. Lang}} \bibinfo{volume}{1}, \bibinfo{number}{1} (\bibinfo{year}{2022}).
\newblock


\bibitem[Garg et~al\mbox{.}(2023)]%
        {garg2023experimenting}
\bibfield{author}{\bibinfo{person}{Sanjam Garg}, \bibinfo{person}{Aarushi Goel}, \bibinfo{person}{Somesh Jha}, \bibinfo{person}{Saeed Mahloujifar}, \bibinfo{person}{Mohammad Mahmoody}, \bibinfo{person}{Guru-Vamsi Policharla}, {and} \bibinfo{person}{Mingyuan Wang}.} \bibinfo{year}{2023}\natexlab{}.
\newblock \showarticletitle{Experimenting with Zero-Knowledge Proofs of Training}.
\newblock \bibinfo{journal}{\emph{Cryptology ePrint Archive}} (\bibinfo{year}{2023}).
\newblock


\bibitem[Goertzel et~al\mbox{.}(2017)]%
        {singularitynet}
\bibfield{author}{\bibinfo{person}{Ben Goertzel}, \bibinfo{person}{Simone Giacomelli}, \bibinfo{person}{David Hanson}, \bibinfo{person}{Cassio Pennachin}, {and} \bibinfo{person}{Marco Argentieri}.} \bibinfo{year}{2017}\natexlab{}.
\newblock \showarticletitle{SingularityNET: A decentralized, open market and inter-network for AIs}.
\newblock \bibinfo{journal}{\emph{Thoughts, Theories Stud. Artif. Intell. Res.}} (\bibinfo{year}{2017}).
\newblock


\bibitem[Goldberg(1991)]%
        {goldberg1991every}
\bibfield{author}{\bibinfo{person}{David Goldberg}.} \bibinfo{year}{1991}\natexlab{}.
\newblock \showarticletitle{What every computer scientist should know about floating-point arithmetic}.
\newblock \bibinfo{journal}{\emph{ACM computing surveys (CSUR)}} \bibinfo{volume}{23}, \bibinfo{number}{1} (\bibinfo{year}{1991}), \bibinfo{pages}{5--48}.
\newblock


\bibitem[Goldberg et~al\mbox{.}(2021)]%
        {goldberg2021cairo}
\bibfield{author}{\bibinfo{person}{Lior Goldberg}, \bibinfo{person}{Shahar Papini}, {and} \bibinfo{person}{Michael Riabzev}.} \bibinfo{year}{2021}\natexlab{}.
\newblock \showarticletitle{Cairo--a Turing-complete STARK-friendly CPU architecture}.
\newblock \bibinfo{journal}{\emph{Cryptology ePrint Archive}} (\bibinfo{year}{2021}).
\newblock


\bibitem[Goldreich and Oren(1994)]%
        {goldreich1994definitions}
\bibfield{author}{\bibinfo{person}{Oded Goldreich} {and} \bibinfo{person}{Yair Oren}.} \bibinfo{year}{1994}\natexlab{}.
\newblock \showarticletitle{Definitions and properties of zero-knowledge proof systems}.
\newblock \bibinfo{journal}{\emph{Journal of Cryptology}} \bibinfo{volume}{7}, \bibinfo{number}{1} (\bibinfo{year}{1994}), \bibinfo{pages}{1--32}.
\newblock


\bibitem[Harris and Waggoner(2019)]%
        {microsoft}
\bibfield{author}{\bibinfo{person}{Justin~D Harris} {and} \bibinfo{person}{Bo Waggoner}.} \bibinfo{year}{2019}\natexlab{}.
\newblock \showarticletitle{Decentralized and collaborative AI on blockchain}. In \bibinfo{booktitle}{\emph{2019 IEEE international conference on blockchain (Blockchain)}}. IEEE, \bibinfo{pages}{368--375}.
\newblock


\bibitem[Hauser({[n.\,d.]})]%
        {bsoft}
\bibfield{author}{\bibinfo{person}{John~R. Hauser}.} \bibinfo{year}{[n.\,d.]}\natexlab{}.
\newblock \bibinfo{title}{Berkeley SoftFloat}.
\newblock
\newblock
\newblock
\shownote{\url{http://www.jhauser.us/arithmetic/SoftFloat.html}}.


\bibitem[Hong et~al\mbox{.}(2022)]%
        {hong2022handcrafted}
\bibfield{author}{\bibinfo{person}{Sanghyun Hong}, \bibinfo{person}{Nicholas Carlini}, {and} \bibinfo{person}{Alexey Kurakin}.} \bibinfo{year}{2022}\natexlab{}.
\newblock \showarticletitle{Handcrafted backdoors in deep neural networks}.
\newblock \bibinfo{journal}{\emph{Advances in Neural Information Processing Systems}}  \bibinfo{volume}{35} (\bibinfo{year}{2022}), \bibinfo{pages}{8068--8080}.
\newblock


\bibitem[Horowitz({[n.\,d.]})]%
        {a16z}
\bibfield{author}{\bibinfo{person}{Andreessen Horowitz}.} \bibinfo{year}{[n.\,d.]}\natexlab{}.
\newblock \bibinfo{title}{Navigating the High Cost of AI Compute}.
\newblock
\newblock
\newblock
\shownote{\url{https://a16z.com/navigating-the-high-cost-of-ai-compute/}}.


\bibitem[Kalodner et~al\mbox{.}(2018)]%
        {kalodner2018arbitrum}
\bibfield{author}{\bibinfo{person}{Harry Kalodner}, \bibinfo{person}{Steven Goldfeder}, \bibinfo{person}{Xiaoqi Chen}, \bibinfo{person}{S~Matthew Weinberg}, {and} \bibinfo{person}{Edward~W Felten}.} \bibinfo{year}{2018}\natexlab{}.
\newblock \showarticletitle{Arbitrum: Scalable, private smart contracts}. In \bibinfo{booktitle}{\emph{27th USENIX Security Symposium (USENIX Security 18)}}. \bibinfo{pages}{1353--1370}.
\newblock


\bibitem[Kurtulmus and Daniel(2018)]%
        {algorithmia}
\bibfield{author}{\bibinfo{person}{A~Besir Kurtulmus} {and} \bibinfo{person}{Kenny Daniel}.} \bibinfo{year}{2018}\natexlab{}.
\newblock \showarticletitle{Trustless machine learning contracts; evaluating and exchanging machine learning models on the ethereum blockchain}.
\newblock \bibinfo{journal}{\emph{arXiv preprint arXiv:1802.10185}} (\bibinfo{year}{2018}).
\newblock


\bibitem[Labs({[n.\,d.]})]%
        {zkmodulus}
\bibfield{author}{\bibinfo{person}{Modulus Labs}.} \bibinfo{year}{[n.\,d.]}\natexlab{}.
\newblock \bibinfo{title}{The Cost of Intelligence}.
\newblock
\newblock
\newblock
\shownote{\url{https://medium.com/@ModulusLabs/chapter-5-the-cost-of-intelligence-da26dbf93307/}}.


\bibitem[Li et~al\mbox{.}(2022)]%
        {li2022smartvm}
\bibfield{author}{\bibinfo{person}{Tao Li}, \bibinfo{person}{Yaozheng Fang}, \bibinfo{person}{Ye Lu}, \bibinfo{person}{Jinni Yang}, \bibinfo{person}{Zhaolong Jian}, \bibinfo{person}{Zhiguo Wan}, {and} \bibinfo{person}{Yusen Li}.} \bibinfo{year}{2022}\natexlab{}.
\newblock \showarticletitle{SmartVM: A Smart Contract Virtual Machine for Fast On-Chain DNN Computations}.
\newblock \bibinfo{journal}{\emph{IEEE Transactions on Parallel and Distributed Systems}} \bibinfo{volume}{33}, \bibinfo{number}{12} (\bibinfo{year}{2022}), \bibinfo{pages}{4100--4116}.
\newblock


\bibitem[Liu et~al\mbox{.}(2021)]%
        {liu2021zkcnn}
\bibfield{author}{\bibinfo{person}{Tianyi Liu}, \bibinfo{person}{Xiang Xie}, {and} \bibinfo{person}{Yupeng Zhang}.} \bibinfo{year}{2021}\natexlab{}.
\newblock \showarticletitle{ZkCNN: Zero knowledge proofs for convolutional neural network predictions and accuracy}. In \bibinfo{booktitle}{\emph{Proceedings of the 2021 ACM SIGSAC Conference on Computer and Communications Security}}. \bibinfo{pages}{2968--2985}.
\newblock


\bibitem[Luu et~al\mbox{.}(2015)]%
        {luu2015demystifying}
\bibfield{author}{\bibinfo{person}{Loi Luu}, \bibinfo{person}{Jason Teutsch}, \bibinfo{person}{Raghav Kulkarni}, {and} \bibinfo{person}{Prateek Saxena}.} \bibinfo{year}{2015}\natexlab{}.
\newblock \showarticletitle{Demystifying incentives in the consensus computer}. In \bibinfo{booktitle}{\emph{Proceedings of the 22Nd acm sigsac conference on computer and communications security}}. \bibinfo{pages}{706--719}.
\newblock


\bibitem[Mamageishvili and Felten(2023)]%
        {mamageishvili2023incentive}
\bibfield{author}{\bibinfo{person}{Akaki Mamageishvili} {and} \bibinfo{person}{Edward~W Felten}.} \bibinfo{year}{2023}\natexlab{}.
\newblock \showarticletitle{Incentive Schemes for Rollup Validators}.
\newblock \bibinfo{journal}{\emph{arXiv preprint arXiv:2308.02880}} (\bibinfo{year}{2023}).
\newblock


\bibitem[Nakamoto(2008)]%
        {nakamoto2008bitcoin}
\bibfield{author}{\bibinfo{person}{Satoshi Nakamoto}.} \bibinfo{year}{2008}\natexlab{}.
\newblock \showarticletitle{Bitcoin: A peer-to-peer electronic cash system}.
\newblock \bibinfo{journal}{\emph{Decentralized business review}} (\bibinfo{year}{2008}).
\newblock


\bibitem[OpenAI({[n.\,d.]})]%
        {openai}
\bibfield{author}{\bibinfo{person}{OpenAI}.} \bibinfo{year}{[n.\,d.]}\natexlab{}.
\newblock \bibinfo{title}{OpenAI's Pricing}.
\newblock
\newblock
\newblock
\shownote{\url{https://openai.com/pricing}}.


\bibitem[Optimism({[n.\,d.]})]%
        {optimism}
\bibfield{author}{\bibinfo{person}{Optimism}.} \bibinfo{year}{[n.\,d.]}\natexlab{}.
\newblock \bibinfo{title}{Ethereum, scaled. OP Mainnet is a low-cost and lightning-fast Ethereum L2 blockchain powered by Optimism.}
\newblock
\newblock
\newblock
\shownote{\url{https://www.optimism.io/}}.


\bibitem[Paszke et~al\mbox{.}(2019)]%
        {paszke2019pytorch}
\bibfield{author}{\bibinfo{person}{Adam Paszke}, \bibinfo{person}{Sam Gross}, \bibinfo{person}{Francisco Massa}, \bibinfo{person}{Adam Lerer}, \bibinfo{person}{James Bradbury}, \bibinfo{person}{Gregory Chanan}, \bibinfo{person}{Trevor Killeen}, \bibinfo{person}{Zeming Lin}, \bibinfo{person}{Natalia Gimelshein}, \bibinfo{person}{Luca Antiga}, {et~al\mbox{.}}} \bibinfo{year}{2019}\natexlab{}.
\newblock \showarticletitle{Pytorch: An imperative style, high-performance deep learning library}.
\newblock \bibinfo{journal}{\emph{Advances in neural information processing systems}}  \bibinfo{volume}{32} (\bibinfo{year}{2019}).
\newblock


\bibitem[Roughgarden(2010)]%
        {roughgarden2010algorithmic}
\bibfield{author}{\bibinfo{person}{Tim Roughgarden}.} \bibinfo{year}{2010}\natexlab{}.
\newblock \showarticletitle{Algorithmic game theory}.
\newblock \bibinfo{journal}{\emph{Commun. ACM}} \bibinfo{volume}{53}, \bibinfo{number}{7} (\bibinfo{year}{2010}), \bibinfo{pages}{78--86}.
\newblock


\bibitem[Salah et~al\mbox{.}(2019)]%
        {salah2019blockchain}
\bibfield{author}{\bibinfo{person}{Khaled Salah}, \bibinfo{person}{M~Habib~Ur Rehman}, \bibinfo{person}{Nishara Nizamuddin}, {and} \bibinfo{person}{Ala Al-Fuqaha}.} \bibinfo{year}{2019}\natexlab{}.
\newblock \showarticletitle{Blockchain for AI: Review and open research challenges}.
\newblock \bibinfo{journal}{\emph{IEEE Access}}  \bibinfo{volume}{7} (\bibinfo{year}{2019}), \bibinfo{pages}{10127--10149}.
\newblock


\bibitem[Sanders and Kandrot(2010)]%
        {sanders2010cuda}
\bibfield{author}{\bibinfo{person}{Jason Sanders} {and} \bibinfo{person}{Edward Kandrot}.} \bibinfo{year}{2010}\natexlab{}.
\newblock \bibinfo{booktitle}{\emph{CUDA by example: an introduction to general-purpose GPU programming}}.
\newblock \bibinfo{publisher}{Addison-Wesley Professional}.
\newblock


\bibitem[Song and Ying(2015)]%
        {song2015decision}
\bibfield{author}{\bibinfo{person}{Yan-Yan Song} {and} \bibinfo{person}{LU Ying}.} \bibinfo{year}{2015}\natexlab{}.
\newblock \showarticletitle{Decision tree methods: applications for classification and prediction}.
\newblock \bibinfo{journal}{\emph{Shanghai archives of psychiatry}} \bibinfo{volume}{27}, \bibinfo{number}{2} (\bibinfo{year}{2015}), \bibinfo{pages}{130}.
\newblock


\bibitem[Soucy and Mineau(2001)]%
        {soucy2001simple}
\bibfield{author}{\bibinfo{person}{Pascal Soucy} {and} \bibinfo{person}{Guy~W Mineau}.} \bibinfo{year}{2001}\natexlab{}.
\newblock \showarticletitle{A simple KNN algorithm for text categorization}. In \bibinfo{booktitle}{\emph{Proceedings 2001 IEEE international conference on data mining}}. IEEE, \bibinfo{pages}{647--648}.
\newblock


\bibitem[Szydlo(2004)]%
        {mt}
\bibfield{author}{\bibinfo{person}{Michael Szydlo}.} \bibinfo{year}{2004}\natexlab{}.
\newblock \showarticletitle{Merkle tree traversal in log space and time}. In \bibinfo{booktitle}{\emph{Advances in Cryptology-EUROCRYPT 2004: International Conference on the Theory and Applications of Cryptographic Techniques, Interlaken, Switzerland, May 2-6, 2004. Proceedings 23}}. Springer, \bibinfo{pages}{541--554}.
\newblock


\bibitem[Teutsch and Reitwie{\ss}ner(2019)]%
        {truebit}
\bibfield{author}{\bibinfo{person}{Jason Teutsch} {and} \bibinfo{person}{Christian Reitwie{\ss}ner}.} \bibinfo{year}{2019}\natexlab{}.
\newblock \showarticletitle{A scalable verification solution for blockchains}.
\newblock \bibinfo{journal}{\emph{arXiv preprint arXiv:1908.04756}} (\bibinfo{year}{2019}).
\newblock


\bibitem[Wang and Hoang(2022)]%
        {wang2022ezdps}
\bibfield{author}{\bibinfo{person}{Haodi Wang} {and} \bibinfo{person}{Thang Hoang}.} \bibinfo{year}{2022}\natexlab{}.
\newblock \showarticletitle{ezDPS: An Efficient and Zero-Knowledge Machine Learning Inference Pipeline}.
\newblock \bibinfo{journal}{\emph{arXiv preprint arXiv:2212.05428}} (\bibinfo{year}{2022}).
\newblock


\bibitem[Weng et~al\mbox{.}(2021)]%
        {weng2021mystique}
\bibfield{author}{\bibinfo{person}{Chenkai Weng}, \bibinfo{person}{Kang Yang}, \bibinfo{person}{Xiang Xie}, \bibinfo{person}{Jonathan Katz}, {and} \bibinfo{person}{Xiao Wang}.} \bibinfo{year}{2021}\natexlab{}.
\newblock \showarticletitle{Mystique: Efficient conversions for $\{$Zero-Knowledge$\}$ proofs with applications to machine learning}. In \bibinfo{booktitle}{\emph{30th USENIX Security Symposium (USENIX Security 21)}}. \bibinfo{pages}{501--518}.
\newblock


\bibitem[Wood et~al\mbox{.}(2014)]%
        {wood2014ethereum}
\bibfield{author}{\bibinfo{person}{Gavin Wood} {et~al\mbox{.}}} \bibinfo{year}{2014}\natexlab{}.
\newblock \showarticletitle{Ethereum: A secure decentralised generalised transaction ledger}.
\newblock \bibinfo{journal}{\emph{Ethereum project yellow paper}} \bibinfo{volume}{151}, \bibinfo{number}{2014} (\bibinfo{year}{2014}), \bibinfo{pages}{1--32}.
\newblock


\bibitem[W{\"u}st et~al\mbox{.}(2020)]%
        {ace}
\bibfield{author}{\bibinfo{person}{Karl W{\"u}st}, \bibinfo{person}{Sinisa Matetic}, \bibinfo{person}{Silvan Egli}, \bibinfo{person}{Kari Kostiainen}, {and} \bibinfo{person}{Srdjan Capkun}.} \bibinfo{year}{2020}\natexlab{}.
\newblock \showarticletitle{ACE: Asynchronous and concurrent execution of complex smart contracts}. In \bibinfo{booktitle}{\emph{Proceedings of the 2020 ACM SIGSAC Conference on Computer and Communications Security}}. \bibinfo{pages}{587--600}.
\newblock


\bibitem[Xing et~al\mbox{.}(2023)]%
        {xing2023zero}
\bibfield{author}{\bibinfo{person}{Zhibo Xing}, \bibinfo{person}{Zijian Zhang}, \bibinfo{person}{Meng Li}, \bibinfo{person}{Jiamou Liu}, \bibinfo{person}{Liehuang Zhu}, \bibinfo{person}{Giovanni Russello}, {and} \bibinfo{person}{Muhammad~Rizwan Asghar}.} \bibinfo{year}{2023}\natexlab{}.
\newblock \showarticletitle{Zero-Knowledge Proof-based Practical Federated Learning on Blockchain}.
\newblock \bibinfo{journal}{\emph{arXiv preprint arXiv:2304.05590}} (\bibinfo{year}{2023}).
\newblock


\bibitem[Yang et~al\mbox{.}(2019)]%
        {yang2019quantization}
\bibfield{author}{\bibinfo{person}{Jiwei Yang}, \bibinfo{person}{Xu Shen}, \bibinfo{person}{Jun Xing}, \bibinfo{person}{Xinmei Tian}, \bibinfo{person}{Houqiang Li}, \bibinfo{person}{Bing Deng}, \bibinfo{person}{Jianqiang Huang}, {and} \bibinfo{person}{Xian-sheng Hua}.} \bibinfo{year}{2019}\natexlab{}.
\newblock \showarticletitle{Quantization networks}. In \bibinfo{booktitle}{\emph{Proceedings of the IEEE/CVF Conference on Computer Vision and Pattern Recognition}}. \bibinfo{pages}{7308--7316}.
\newblock


\bibitem[Yang et~al\mbox{.}(2022)]%
        {yang2022fusing}
\bibfield{author}{\bibinfo{person}{Qinglin Yang}, \bibinfo{person}{Yetong Zhao}, \bibinfo{person}{Huawei Huang}, \bibinfo{person}{Zehui Xiong}, \bibinfo{person}{Jiawen Kang}, {and} \bibinfo{person}{Zibin Zheng}.} \bibinfo{year}{2022}\natexlab{}.
\newblock \showarticletitle{Fusing blockchain and AI with metaverse: A survey}.
\newblock \bibinfo{journal}{\emph{IEEE Open Journal of the Computer Society}}  \bibinfo{volume}{3} (\bibinfo{year}{2022}), \bibinfo{pages}{122--136}.
\newblock


\bibitem[Zero({[n.\,d.]})]%
        {risv0}
\bibfield{author}{\bibinfo{person}{RISC Zero}.} \bibinfo{year}{[n.\,d.]}\natexlab{}.
\newblock \bibinfo{title}{RISC Zero : General-Purpose Verifiable Computing}.
\newblock
\newblock
\newblock
\shownote{\url{https://www.risczero.com/}}.


\bibitem[Zhang et~al\mbox{.}(2020)]%
        {zhang2020zero}
\bibfield{author}{\bibinfo{person}{Jiaheng Zhang}, \bibinfo{person}{Zhiyong Fang}, \bibinfo{person}{Yupeng Zhang}, {and} \bibinfo{person}{Dawn Song}.} \bibinfo{year}{2020}\natexlab{}.
\newblock \showarticletitle{Zero knowledge proofs for decision tree predictions and accuracy}. In \bibinfo{booktitle}{\emph{Proceedings of the 2020 ACM SIGSAC Conference on Computer and Communications Security}}. \bibinfo{pages}{2039--2053}.
\newblock


\bibitem[Zheng et~al\mbox{.}(2021)]%
        {zheng2021agatha}
\bibfield{author}{\bibinfo{person}{Zihan Zheng}, \bibinfo{person}{Peichen Xie}, \bibinfo{person}{Xian Zhang}, \bibinfo{person}{Shuo Chen}, \bibinfo{person}{Yang Chen}, \bibinfo{person}{Xiaobing Guo}, \bibinfo{person}{Guangzhong Sun}, \bibinfo{person}{Guangyu Sun}, {and} \bibinfo{person}{Lidong Zhou}.} \bibinfo{year}{2021}\natexlab{}.
\newblock \showarticletitle{Agatha: Smart contract for DNN computation}.
\newblock \bibinfo{journal}{\emph{arXiv preprint arXiv:2105.04919}} (\bibinfo{year}{2021}).
\newblock


\end{thebibliography}

\end{document}